\documentclass[UKenglish]{llncs}

\newif\iflong 
\longtrue 

\iflong
\title{A Generic Logic for Proving Linearizability\\(Extended Version)}
\else
\title{A Generic Logic for Proving Linearizability}
\fi
\titlerunning{A Generic Logic for Proving Linearizability}

\author{Artem Khyzha\inst{1} \and
Alexey Gotsman\inst{1} \and
Matthew Parkinson\inst{2}}

\institute{IMDEA Software Institute
\and Microsoft Research Cambridge}

\usepackage[frame,all]{xy}
\usepackage{amsmath}
\usepackage{amssymb}
\usepackage{bussproofs}
\usepackage{color}
\usepackage{etoolbox}
\usepackage{extarrows}
\usepackage{framed}
\usepackage{lipsum}
\usepackage{listings}
\usepackage{mathrsfs}
\usepackage{mathtools}
\usepackage{placeins}
\usepackage{stmaryrd}
\usepackage{myarrows}
\usepackage{wrapfig}
\usepackage{xspace}
\usepackage{yfonts}
\usepackage{float}

\newcommand{\myfrac}[2]{\genfrac{}{}{0.5pt}{}{\displaystyle #1}{\displaystyle #2}}
\newcommand{\mypar}[1]{\vspace{5pt}\noindent\textbf{#1.}}
\newcommand{\nlfrac}[2]{\genfrac{}{}{0pt}{}{\displaystyle #1}{\displaystyle #2}}

\newcommand{\sothat}{{.\,}}

\newtheorem{thm}{Theorem}
\newtheorem{prop}[thm]{Proposition}

\newtheorem{lem}[thm]{Lemma}
\newtheorem{dfn}[thm]{Definition}

\DeclarePairedDelimiter{\da}{\langle}{\rangle}
\DeclarePairedDelimiter{\db}{\llbracket}{\rrbracket}
\DeclarePairedDelimiter{\dc}{\lbrace}{\rbrace}
\DeclarePairedDelimiter{\df}{\lfloor}{\rfloor}
\AtBeginDocument{ \addtolength\abovedisplayskip{-0.11\baselineskip}
  \addtolength\belowdisplayskip{-0.12\baselineskip}
  \addtolength\abovedisplayshortskip{-0.1\baselineskip}
  \addtolength\belowdisplayshortskip{-0.1\baselineskip} }

\newcommand{\powerset}[1]{\mathcal{P}(#1)}

\newcommand{\HPCom}{{\sf APCom}}

\newcommand{\HState}{{\sf AState}}

\newcommand{\Int}{{\sf Int}}
\newcommand{\LCom}{{\sf Com}}
\newcommand{\Loc}{{\sf Loc}}
\newcommand{\LPCom}{{\sf PCom}}
\newcommand{\LState}{{\sf State}}

\newcommand{\LVars}{{\sf LVar}}
\newcommand{\RGsep}{{\sf RGsep}}
\newcommand{\ThreadID}{{\sf ThreadID}}
\newcommand{\Vals}{{\sf Val}}

\newcommand{\cskip}{{\sf skip}}
\newcommand{\emptymap}{[\ ]}
\newcommand{\emptytr}{\varepsilon}

\newcommand{\hpcom}{{\rm A}}
\newcommand{\hstate}{\Sigma}
\newcommand{\id}{{\sf id}}
\newcommand{\lcom}{C}

\newcommand{\lpcom}{{\sf \alpha}}
\newcommand{\lP}{\mathcal{P}}
\newcommand{\lQ}{\mathcal{Q}}
\newcommand{\lR}{\mathcal{R}}
\newcommand{\lstate}{\sigma}
\DeclareTextFontCommand{\textbfit}{
  \fontseries\bfdefault 
  \itshape
}
\makeatletter\def\@listI{\leftmargin\leftmargini\parsep 1\p@ \@plus1\p@
\@minus\p@\topsep 1\p@ \@plus2\p@ \@minus0\p@\itemsep0\p@}\let\@listi
\@listI\@listi\makeatother

\newcommand{\stable}{{\sf stable}}
\newcommand{\dom}{{\sf dom}}
\newcommand{\evalf}[2]{\db{#1}_{#2}}
\newcommand{\reif}[1]{\df{{#1}}}

\newcommand{\aslto}{\mathbin{\Mapsto}}
\newcommand{\cc}{\mathbin{{:}{:}}}
\newcommand{\ldet}{\mathbin{+}}
\newcommand{\lkstar}{\star}

\newcommand{\lseq}{\mathbin{\,;\,}}
\newcommand{\pto}{\rightharpoonup}
\newcommand{\slto}{\mathbin{\mapsto}}
\newcommand{\vop}{\mathbin{*}}
\newcommand{\vopall}{\circledast}

\newcommand{\lstassert}[1]{{\color{blue}\ensuremath{\left\lbrace \begin{array}{l}#1\end{array} \right\rbrace}}}

\begin{document}
\maketitle

\begin{abstract}
  Linearizability is a commonly accepted notion of correctness for libraries of
  concurrent algorithms, and recent years have seen a number of proposals of
  program logics for proving it. Although these logics differ in technical
  details, they embody similar reasoning principles. To explicate these
  principles, we propose a logic for proving linearizability that is generic: it
  can be instantiated with different means of compositional reasoning about
  concurrency, such as separation logic or rely-guarantee. To this end, we
  generalise the Views framework for reasoning about concurrency to handle
  relations between programs, required for proving linearizability. We present
  sample instantiations of our generic logic and show that it is powerful enough
  to handle concurrent algorithms with challenging features, such as helping.
\end{abstract}

\section{Introduction}

To manage the complexity of constructing concurrent software, programmers
package often-used functionality into {\em libraries} of concurrent algorithms.
These encapsulate data structures, such as queues and lists, and provide clients
with a set of methods that can be called concurrently to operate on these
(e.g.,~\textsf{java.util.concurrent}). To maximise performance, concurrent
libraries may use sophisticated non-blocking techniques, allowing multiple
threads to operate on the data structure with minimum synchronisation.  Despite
this, each library method is usually expected to behave as though it executes
atomically. This requirement is formalised by the standard notion of correctness
for concurrent libraries, {\em linearizability}~\cite{linearizability}, which
establishes a form of a simulation between the original {\em concrete} library
and another {\em abstract} library, where each method is implemented atomically.

A common approach to proving linearizability is to find a {\em linearization
  point} for every method of the concrete library at which it can be thought of
  taking effect.
\footnote{Some algorithms cannot be reasoned about using linearization points,
which we discuss in \S\ref{sec:related}.}
Given an execution of a concrete library, the matching execution of the abstract
library, required to show the simulation, is constructed by executing the atomic
abstract method at the linearization point of the concrete method. A difficulty
in this approach is that linearization points are often not determined by a
statically chosen point in the method code. For example, in concurrent
algorithms with {\em helping}~\cite{herlihy-book}, a method may execute an
operation originally requested by another method, called in a different thread;
then the linearization point of the latter method is determined by an action of
the former.

Recent years have seen a number of program logics for proving linearizability
(see~\cite{DongolD14} for a survey). To avoid reasoning about the high number
of possible interleavings between concurrently executing threads, these logics
often use {\em thread-modular} reasoning. They establish protocols that
threads should follow when operating on the shared data structure and reason
separately about every thread, assuming that the rest follow the protocols. The
logics for proving linearizability, such as~\cite{rgsep,rgsim-lin}, usually
borrow thread-modular reasoning rules from logics originally designed for
proving non-relational properties of concurrent programs, such as
rely-guarantee~\cite{rg}, separation logic~\cite{seplogic} or combinations
thereof~\cite{rgsep,lrg}. Although this leads the logics to differ in technical
details, they use similar methods for reasoning about linearizability, usually
based on linearization points. Despite this similarity, designing a logic for
proving linearizability that uses a particular thread-modular reasoning method
currently requires finding the proof rules and proving their soundness afresh.

To consolidate this design space of linearization-point-based reasoning, we
propose a logic for linearizability that is {\em generic}, i.e., can be
instantiated with different means of thread-modular reasoning about concurrency,
such as separation logic~\cite{seplogic} or rely-guarantee~\cite{rg}. To this
end, we build on the recently-proposed Views framework~\cite{views}, which
unifies thread-modular logics for concurrency, such as the above-mentioned
ones. Our contribution is to generalise the framework to reason about relations
between programs, required for proving linearizability. In more detail,
assertions in our logic are interpreted over a monoid of {\em relational views},
which describe relationships between the states of the concrete and the abstract
libraries and the protocol that threads should follow in operating on these. The
operation of the monoid, similar to the separating conjunction in separation
logic~\cite{seplogic}, combines the assertions in different threads while
ensuring that they agree on the protocols of access to the state. The choice of
a particular relational view monoid thus determines the thread-modular reasoning
method used by our logic.

To reason about linearization points, relational views additionally describe a
set of special {\em tokens} (as in~~\cite{rgsep,rgsim-lin,tada}), each denoting
a one-time permission to execute a given atomic command on the state of the
abstract library. The place where this permission is used in the proof of a
concrete library method determines its linearization point, with the abstract
command recorded by the token giving its specification. Crucially, reasoning
about the tokens is subject to the protocols established by the underlying
thread-modular reasoning method; in particular, their \emph{ownership} can be
transferred between different threads, which allows us to deal with helping.

We prove the soundness of our generic logic under certain conditions on its
instantiations (Definition~\ref{def:axiom}, \S\ref{sec:logic}). These conditions
represent our key technical contribution, as they capture the essential
requirements for soundly combining a given thread-modular method for reasoning
about concurrency with the linearization-point method for reasoning about
linearizability.

To illustrate the use of our logic, we present its example instantiations where
thread-modular reasoning is done using disjoint concurrent separation
logic~\cite{csl} and a combination of separation logic and
rely-guarantee~\cite{rgsep}. We then apply the latter instantiation to prove
the correctness of a sample concurrent algorithm with helping. We expect that
our results will make it possible to systematically design logics using the
plethora of other methods for thread-modular reasoning that have been shown to
be expressible in the Views framework~\cite{csl,CAP,seplogicperm}.

\FloatBarrier
\newcommand{\NThreads}{N}
\newcommand{\assume}{{\tt assume}}
\newcommand{\tid}{t}

\newcommand{\intp}[3]{\db{#2}_{#1}(#3)}
\newcommand{\trans}[3]{#1 \mathrel{{\xlongrightarrow{#2}}{}} #3}
\newcommand{\sttrans}[6]{\da{#1,#2} \mathrel{{\xtwoheadrightarrow{#3, #4}}{}} \da{#5,#6}}

\newcommand{\iter}[1]{{#1}^\lkstar}

\newcommand{\llib}{\mathrm{\ell}}
\newcommand{\hlib}{\mathcal{L}}

\section{Methods Syntax and Sequential Semantics}\label{sec:proglang}

We consider concurrent programs that consist of two components, which we call
{\em libraries} and {\em clients}. Libraries provide clients with a set of
methods, and clients call them concurrently. We distinguish {\em concrete} and
{\em abstract} libraries, as the latter serve as specification for the
former due to its methods being executed atomically.


\mypar{Syntax}
 Concrete methods are
  implemented as {\em sequential commands} having the syntax:
\[
\lcom \in \LCom ::= \lpcom \mid \lcom \lseq \lcom \mid \lcom \ldet
\lcom \mid \iter{\lcom} \mid \cskip,\quad \mbox{where } \lpcom \in \LPCom
\]
The grammar includes {\em primitive commands} $\lpcom$ from a set $\LPCom$,
sequential composition $\lcom \lseq \lcom$, non-deterministic choice $\lcom
\ldet \lcom$ and a finite iteration $\iter{\lcom}$ (we are interested only in
terminating executions) and a termination marker $\cskip$. We use $\ldet$ and
$(\cdot)^\lkstar$ instead of conditionals and while loops for theoretical
simplicity: as we show at the end of this section, given appropriate primitive
commands the conditionals and loops can be encoded. We also assume a set
$\HPCom$ of {\em abstract primitive commands}, ranged over by $\hpcom$, with
which we represent methods of an abstract library.

\begin{figure}[t]
${\longrightarrow} \subseteq \LCom \times \LPCom \times \LCom$

{\centering
$\begin{array}{l@{\ }l@{\ }l}
\myfrac{
  \trans{\lcom_1}{\lpcom}{\lcom'_1}
}{
  \trans{\lcom_1 \lseq \lcom_2}{\lpcom}{\lcom'_1 \lseq \lcom_2}
} &
\myfrac{
  \,
}{
  \trans{\iter{\lcom}}{\id}{\lcom; \iter{\lcom}}
} &
\myfrac{
  \,
}{
  \trans{\lpcom}{\lpcom}{\cskip}
} \\
\myfrac{
  i \in \{1, 2\}
}{
  \trans{\lcom_1 \ldet \lcom_2}{\id}{\lcom_i}
} &
\myfrac{
  \,
}{
  \trans{\cskip \lseq \lcom}{\id}{\lcom}
}
&
\myfrac{
  \,
}{
  \trans{\iter{\lcom}}{\id}{\cskip}
}
\end{array}
$\par}

\medskip

${\xtwoheadrightarrow{}} \subseteq (\LCom \times \LState) \times (\ThreadID \times \LPCom) \times (\LCom \times \LState)$

{\centering
$\myfrac{
  \lstate' \in \intp{\tid}{\lpcom}{\lstate} \quad \trans{\lcom}{\lpcom}{\lcom'}
}{
  \sttrans{\lcom}{\lstate}{\tid}{\lpcom}{\lcom'}{\lstate'}
}$\par}
\caption{The operational semantics of sequential commands\label{fig:opsem}}
\end{figure}

\mypar{Semantics} We assume a set $\LState$ of concrete states of the memory,
ranged over by $\lstate$, and abstract states $\HState$, ranged over by
$\hstate$. The memory is shared among $\NThreads$ threads with thread
identifiers $\ThreadID = \{1, 2, \dots, \NThreads\}$, ranged over by $\tid$.

We assume that semantics of each primitive command $\lpcom$ is given by a
non-deterministic {\em state transformer} $\db{\lpcom}_\tid : \LState \to
\powerset{\LState}$, where $\tid \in \ThreadID$. For a state $\lstate$, the set
of states $\intp{\tid}{\lpcom}{\lstate}$ is the set of possible resulting
states for $\lpcom$ executed atomically in a state $\lstate$ and a thread
$\tid$. State transformers may have different semantics depending on a thread
identifier, which we use to introduce thread-local memory cells later in the
technical development. Analogously, we assume semantics of abstract primitive
commands with state transformers $\db{\hpcom}_\tid : \HState \to
\powerset{\HState}$, all of which update abstract states atomically. We also
assume a primitive command $\id \in \LPCom$ with the interpretation
$\intp{\tid}{\id}{\lstate} \triangleq \{\lstate \}$, and its abstract
counterpart $\id \in \HPCom$.

The sets of primitive commands $\LPCom$ and $\HPCom$ as well as corresponding
state transformers are parameters of our framework. In Figure~\ref{fig:opsem} we
give rules of operational semantics of sequential commands, which are
parametrised by semantics of primitive commands. That is, we define a transition
relation ${\xtwoheadrightarrow{}}~\subseteq~(\LCom \times \LState)
\times (\ThreadID \times \LPCom) \times (\LCom \times \LState)$, so that
$\sttrans{\lcom}{\lstate}{\tid}{\lpcom}{\lcom'}{\lstate'}$ indicates a
transition from $\lcom$ to $\lcom'$ updating the state from $\lstate$ to
$\lstate'$ with a primitive command $\lpcom$ in a thread $\tid$. The rules of
the operational semantics are standard.

Let us show how to define traditional control flow primitives, such as an
if-statement and a while-loop, in our programming language. Assuming a language
for arithmetic expressions, ranged over by $E$, and a function $\db{E}_\lstate$
that evaluates expressions in a given state $\lstate$, we define a
primitive command $\assume(E)$ that acts as a filter on states, choosing only
those where $E$ evaluates to non-zero values.
\[
\intp{\tid}{\assume(E)}{\lstate} \triangleq
  \mbox{ if } \db{E}_\lstate \not= 0
  \mbox{ then } \dc{\lstate}
  \mbox{ else } \emptyset.
\]
Using $\assume(E)$ and the C-style negation $!E$ in expressions, a conditional
and a while-loop can be implemented as the following commands:
\begin{gather*}
  {\tt if}\ E\ {\tt then}\ C_1\ {\tt else}\ C_2 \triangleq (\assume(E); C_1) \ldet (\assume(!E); C_2) \\
  {\tt while}\ E\ {\tt do}\ C \triangleq (\assume(E); C)^\lkstar; \assume(!E)
\end{gather*}


\FloatBarrier

\newcommand{\done}[1]{{\sf done}(#1)}
\newcommand{\ldone}[2]{\left [{\sf done}(#2) \right ]_{#1}}
\newcommand{\lint}{{\bf i}}
\newcommand{\ltodo}[2]{\left [{\sf todo}(#2) \right ]_{#1}}
\newcommand{\todo}[1]{{\sf todo}(#1)}
\newcommand{\tstate}{\Delta}
\newcommand{\unit}{u}
\newcommand{\fault}{\lightning}
\newcommand{\vassn}{\rho}

\newcommand{\rimpl}{\Rrightarrow}
\newcommand{\req}{\mathrlap{\Lleftarrow\ }{\ \Rrightarrow}}

\newcommand{\Assn}{{\sf Assn}}
\newcommand{\Axioms}{{\sf Axioms}}
\newcommand{\Tokens}{{\sf Tokens}}
\newcommand{\VAssn}{{\sf VAssn}}
\newcommand{\Views}{{\sf Views}}

\newcommand{\slop}{\bullet}

\newcommand{\infer}[4]{\vdash_{#1} \left\{{#2}\right\} \,#3\, \left\{{#4}\right\}}
\newcommand{\ptofin}{\pto_{\rm fin}}
\newcommand{\safe}[4]{{\sf safe}_{#1}(#2, #3, #4)}
\newcommand{\sat}[4]{\,#2 \Vdash_{#1} \dc{#3}\dc{#4}}
\newcommand{\lp}[2]{{\sf LP}(#1, #2)}
\newcommand{\lptrans}[2]{{\sf LP}^*(#1, #2)}

\newcommand{\assert}[3]{{#1} = \evalf{{#3}}{{#2}}}

\section{The generic logic}\label{sec:logic}

In this section, we present our framework for designing program logics for
linearizability proofs. Given a concrete method and a corresponding abstract
method, we aim to demonstrate that the former has a linearization point either
within its code or in the code of another thread. The idea behind such proofs
is to establish simulation between concrete and abstract methods using
linearization points to determine when the abstract method has to make a
transition to match a given execution of the concrete method. To facilitate
such simulation-based proofs, we design our relational logic so that formulas
in it denote relations between concrete states, abstract states and special
{\it tokens}.

Tokens are our tool for reasoning about linearization points. At the beginning
of its execution in a thread $\tid$, each concrete method $m$ is given a token
$\todo{\hpcom_m}$ of the corresponding abstract primitive command $\hpcom_m$.
The token represents a one-time permission for the method to take effect, i.e.
to perform a primitive command $\hpcom_m$ on an abstract machine. When the
permission is used, a token $\todo{\hpcom_m}$ in a thread $\tid$ is irreversibly
replaced with $\done{\hpcom_m}$. Thus, by requiring that a method start its
execution in a thread $\tid$ with a token $\todo{\hpcom_m}$ and ends with
$\done{\hpcom_m}$, we ensure that it has in its code a linearization point. The
tokens of all threads are described by $\tstate \in \Tokens$:
\[
\Tokens = \ThreadID \pto (\dc{ \todo{\hpcom} \mid \hpcom \in \HPCom} \cup
  \dc{ \done{\hpcom} \mid \hpcom \in \HPCom})
\]

Reasoning about states and tokens in the framework is done with the help of
relational {\it views}. We assume a set $\Views$, ranged over by $p$, $q$ and
$r$, as well as a reification function $\reif{\ } : \Views \to \powerset{\LState
\times \HState \times \Tokens}$ that interprets views as ternary relations on
concrete states, abstract states and indexed sets of tokens.

\begin{dfn}
  A relational view monoid is a commutative monoid $(\Views, \vop, \unit)$,
  where $\Views$ is an underlying set of relational views, $\vop$ is a monoid
  operation and $\unit$ is a unit.
\end{dfn}

The monoid structure of relational views allows treating them as restrictions
on the environment of threads. Intuitively, each thread uses views to declare a
protocol that other threads should follow while operating with concrete states,
abstract states and tokens. Similarly to the separating conjunction from
separation logic, the monoid operation $\vop$ (view composition) applied to a
pair of views combines protocols of access to the state and ensures that they
do not contradict each other.

\mypar{Disjoint Concurrent Separation logic} To give an example of a view
monoid, we demonstrate the structure inspired by Disjoint Concurrent Separation
logic (DCSL). A distinctive feature of DCSL is that its assertions enforce a
protocol, according to which threads operate on disjoint pieces of memory. We
assume a set of values $\Vals$, of which a subset $\Loc \subseteq \Vals$
represents heap addresses. By letting $\LState = \HState = (\Loc \ptofin
\Vals) \cup \dc{\fault}$ we represent a state as either a finite partial
function from locations to values or an exceptional faulting state $\fault$,
which denotes the result of an invalid memory access.
We define an operation $\slop$ on states, which results in $\fault$ if either of
the operands is $\fault$, or the union of partial functions if their domains are
disjoint. Finally, we assume that the set $\LPCom$ consists of standard
heap-manipulating commands with usual semantics \cite{seplogic,views}.

We consider the view monoid $(\powerset{(\LState \setminus \dc{\fault}) \times
(\HState \setminus \dc{\fault}) \times \Tokens}, \vop_{\sf SL}, (\emptymap,
\emptymap, \emptymap))$: the unit is a triple of nowhere defined functions
$\emptymap$, and the view composition defined as follows:
\[
p \vop_{\sf SL} p' \triangleq 
\dc{(\lstate \slop \lstate', \hstate \slop \hstate', \tstate \uplus \tstate')
    \mid (\lstate, \hstate, \tstate) \in p 
    \land (\lstate', \hstate', \tstate') \in p'}.
\]
In this monoid, the composition enforces a protocol of exclusive ownership of
parts of the heap: a pair of views can be composed only if they do not
simultaneously describe the content of the same heap cell or a token. Since
tokens are exclusively owned in DCSL, they cannot be accessed by other threads,
which makes it impossible to express a helping mechanism with the DCSL views.
In \S\ref{sec:rgsep}, we present another instance of our framework and
reason about helping in it.

\mypar{Reasoning about linearization points}
We now introduce {\em action judgements}, which formalise
linearization-points-based approach to proving linearizability within our
framework.

\begin{wrapfigure}[5]{r}{8em}
\vspace{-35pt}
\begin{center}
$
\xymatrix @R=1.2em @C=1.5em {
\lstate \ar@{-}[rr]^{\reif{p \vop r}} \ar[dd]^{\db{\lpcom}} &&
\hstate \ar@{-->}[dd]^{\db{\hpcom}} \\ \\
\lstate' \ar@{--}[rr]^{\reif{q \vop r}} &&
\hstate'
}
$
\end{center}
\end{wrapfigure}
Let us assume that $\lpcom$ is executed in a concrete state $\lstate$ with an
abstract state $\hstate$ and a set of tokens $\tstate$ satisfying a precondition
$p$. According to the action judgement $\sat{\tid}{\lpcom}{p}{q}$, for every
update $\lstate' \in \intp{\tid}{\lpcom}{\lstate}$ of the concrete state, the
abstract state may be changed to $\hstate' \in
\intp{\tid'}{\hpcom}{\hstate}$ in order to satisfy the postcondition $q$,
provided that there is a token $\todo{\hpcom}$ in a thread $\tid'$. When the
abstract state $\hstate$ is changed and the token $\todo{\hpcom}$ of a thread
$\tid'$ is used, the concrete state update corresponds to a linearization
point, or to a regular transition otherwise.

\begin{dfn}\label{def:axiom} The action judgement $\sat{\tid}{\lpcom}{p}{q}$
holds, iff the following is true:
\begin{multline*}
\forall r, \lstate, \lstate', \hstate, \tstate \sothat 
  (\lstate, \hstate, \tstate) \in \reif{p \vop r} \land
  \lstate' \in \intp{\tid}{\lpcom}{\lstate} \implies{} \\
    \exists \hstate', \tstate' \sothat
      \lptrans{\hstate, \tstate}{\hstate', \tstate'}
      \land (\lstate', \hstate', \tstate') \in \reif{q \vop r},
\end{multline*}
where ${\sf LP}^*$ is the transitive closure of the following relation:
\begin{equation*}
\lp{\hstate, \tstate}{\hstate', \tstate'} \triangleq
\exists \tid', \hpcom \ldotp
\hstate' \in \intp{\tid'}{\hpcom}{\hstate}
\land
\tstate(\tid') = \todo{\hpcom}
\land
\tstate' = \tstate[\tid' : \done{\hpcom}],
\end{equation*}
and $f[x : a]$ denotes the function such that $f[x : a](x) = a$ and for any $y
\neq x$, $f[x : a](y) = f(y)$.
\end{dfn}\FloatBarrier
Note that depending on pre- and postconditions $p$ and $q$,
$\sat{\tid}{\lpcom}{p}{q}$ may encode a regular transition, a conditional or a
standard linearization point. It is easy to see that the latter is the case only
when in all sets of tokens $\tstate$ from $\reif{p}$ some thread $\tid'$ has a
todo-token, and in all $\Delta'$ from $\reif{q}$ it has a done-token.
Additionally, the action judgement may represent a conditional linearization
point of another thread, as the ${\sf LP}$ relation allows using tokens of other
threads.


Action judgements have a closure property that is important for thread-modular
reasoning: when $\sat{\tid}{\lpcom}{p}{q}$ holds, so does $\sat{\tid}{\lpcom}{p
* r}{q * r}$ for every view $r$. That is, execution of $\lpcom$ and a
  corresponding linearization point preserves every view $r$ that $p$ can be
  composed with. Consequently, when in every thread action judgements hold of
  primitive commands and thread's views, all threads together mutually agree on
  each other's protocols of the access to the shared memory encoded in their
  views. This enables reasoning about every thread in isolation with the
  assumption that its environment follows its protocol. Thus, the action
  judgements formalise the requirements that instances of our framework need to
  satisfy in order to be sound. In this regard action judgements are inspired
  by semantic judgements of the Views Framework \cite{views}. Our technical
  contribution is in formulating the essential requirements for thread-modular
  reasoning about linearizability of concurrent libraries with the
  linearization-point method and in extending the semantic judgement with them.

We let a {\it repartitioning implication} of views $p$ and $q$, written $p
\rimpl q$, denote $\forall r \sothat \reif{p * r} \subseteq \reif{q * r}$.
A repartitioning implication $p \rimpl q$ ensures that states satisfying $p$
also satisfy $q$ and additionally requires this property to preserve any view
$r$.

\mypar{Program logic} We are now in a position to present our generic logic for
linearizability proofs via the linearization-point method. Assuming a view
monoid and reification function as parameters, we define a minimal language
$\Assn$ for assertions $\lP$ and $\lQ$ denoting sets of views:
\begin{gather*}
  \lP, \lQ \in \Assn ::= 
    \vassn \mid \lP \vop \lQ \mid \lP \lor \lQ \mid \lP \rimpl \lQ 
    \mid \exists X \sothat \lP \mid \dots
\end{gather*}
The grammar includes view assertions $\vassn$, a syntax $\VAssn$ of which is a
parameter of the framework. Formulas of $\Assn$ may contain the standard
connectives from separation logic, the repartitioning implication and the
existential quantification over logical variables $X$, ranging over a set
$\LVars$.

\begin{figure}[t]
$
\begin{array}{l c r}
\evalf{\lP \vop \lQ}{\lint} = \evalf{\lP}{\lint} \vop \evalf{\lQ}{\lint}
& \quad\quad\quad\quad\quad &
\evalf{\lP \rimpl \lQ}{\lint} = \evalf{\lP}{\lint} \rimpl \evalf{\lQ}{\lint}
\\[1pt]
\evalf{\lP \lor \lQ}{\lint} = \evalf{\lP}{\lint} \lor \evalf{\lQ}{\lint}
& \hfill &
\evalf{\exists X \sothat \lP}{\lint} = \bigvee_{n \in \Vals} \evalf{\lP}{\lint[X : n]}
\end{array}
$
\caption{Satisfaction relation for the assertion language ${\sf Assn}$\label{fig:sem4frm}}
\end{figure}

Let us assume an interpretation of logical variables $\lint \in \Int = \LVars
\to \Vals$ that maps logical variables from $\LVars$ to values from a finite
set $\Vals$. In Figure~\ref{fig:sem4frm}, we define a function
$\db{\cdot}_\cdot : \Assn \times \Int \to \Views$ that we use to interpret
assertions. Interpretation of assertions is parametrised by $\db{\cdot}_\cdot :
\VAssn \times \Int \to \Views$. In order to interpret disjunction, we introduce
a corresponding operation on views and require the following properties from
it:
\begin{equation}\label{eq:disj}
\begin{array}{l@{\hspace{7em}}l}
\reif{p \lor q} = \reif{p} \cup \reif{q}&
(p \lor q) * r = (p * r) \lor (q * r)
\end{array}
\end{equation}

The judgements of the program logic take the form
$\infer{\tid}{\lP}{\lcom}{\lQ}$. In Figure~\ref{fig:proofrules}, we present the
proof rules, which are mostly standard. Among them, the {\sc Prim} rule is
noteworthy, since it encorporates the simulation-based approach to reasoning
about linearization points introduced by action judgements. The {\sc Frame}
rule applies the idea of local reasoning from separation logic \cite{seplogic}
to views. The {\sc Conseq} enables weakening a precondition or a postcondition
in a proof judgement and uses repartitioning implications to ensure the
thread-modularity of the weakened proof judgement.

\begin{figure}[t]
$
\begin{array}{r@{\hspace{2pt}}c@{\hspace{2pt}}r@{\hspace{2pt}}c}
\mbox{\footnotesize\sc(Prim)}\hspace{0.25cm}&
\dfrac{
  \forall \lint \sothat \sat{\tid}{\lpcom}{\evalf{\lP}{\lint}}{\evalf{\lQ}{\lint}}
}{
  \infer{\tid}{\lP}{\lpcom}{\lQ}
} &
\mbox{\footnotesize\sc(Seq)}\hspace{0.25cm}&
\dfrac{
  \infer{\tid}{\lP}{\lcom_1}{\lP'} \quad \infer{\tid}{\lP'}{\lcom_2}{\lQ} 
}{
  \infer{\tid}{\lP}{\lcom_1 \lseq \lcom_2}{\lQ}
}\\[9pt]
\mbox{\footnotesize\sc(Frame)}\hspace{0.25cm}&
\dfrac{
  \infer{\tid}{\lP}{\lcom}{\lQ}
}{
  \infer{\tid}{\lP \vop \lR}{\lcom}{\lQ \vop \lR}
} &
\mbox{\footnotesize\sc(Disj)}\hspace{0.25cm}&
\dfrac{
  \infer{\tid}{\lP_1}{\lcom}{\lQ_1}\quad\infer{\tid}{\lP_2}{\lcom}{\lQ_2}
}{
  \infer{\tid}{\lP_1 \lor \lP_2}{\lcom}{\lQ_1 \lor \lQ_2}
}\\[9pt]
\mbox{\footnotesize\sc(Ex)}\hspace{0.25cm}&
\dfrac{
  \infer{\tid}{\lP}{\lcom}{\lQ}
}{
  \infer{\tid}{\exists X \sothat \lP}{\lcom}{\exists X \sothat \lQ}
} &
\mbox{\footnotesize\sc(Choice)}\hspace{0.25cm}&
\dfrac{
  \infer{\tid}{\lP}{\lcom_1}{\lQ} \quad \infer{\tid}{\lP}{\lcom_2}{\lQ} 
}{
  \infer{\tid}{\lP}{\lcom_1 \ldet \lcom_2}{\lQ}
} \\[9pt]
\mbox{\footnotesize\sc(Iter)}\hspace{0.25cm}&
\dfrac{
  \infer{\tid}{\lP}{\lcom}{\lP}
}{
  \infer{\tid}{\lP}{\lcom^\lkstar}{\lP}
} &
\mbox{\footnotesize\sc(Conseq)}\hspace{0.25cm}&
\dfrac{
  \lP' \Rrightarrow \lP \quad \infer{\tid}{\lP}{\lcom}{\lQ} \quad \lQ \Rrightarrow \lQ'
}{
  \infer{\tid}{\lP'}{\lcom}{\lQ'}
}
\end{array}$
\caption{Proof rules\label{fig:proofrules}}
\end{figure}

\mypar{Semantics of proof judgements}
We give semantics to judgements of the program logic by lifting the
requirements of action judgements to sequential commands.
\begin{dfn}[Safety Judgement]\label{def:safety}
  We define ${\sf safe}_\tid$ as the greatest relation such that the following holds whenever $\safe{\tid}{p}{\lcom}{q}$ does:
  \begin{itemize}
    \item if $\lcom \not= \cskip$, then $\forall \lcom', \lpcom \sothat 
          \trans{\lcom}{\lpcom}{\lcom'} \implies \exists p' \sothat 
          \sat{\tid}{\lpcom}{p}{p'} \land \safe{\tid}{p'}{\lcom'}{q}$,
    \item if $\lcom = \cskip$, then $p \rimpl q$.
  \end{itemize}
\end{dfn}
\begin{lem}\label{lem:logic2safety}
    $\forall \tid, \lP, \lcom, \lQ \sothat \infer{\tid}{\lP}{\lcom}{\lQ} \implies 
        \forall \lint \sothat \safe{\tid}{\evalf{\lP}{\lint}}{\lcom}{\evalf{\lQ}{\lint}}.$
\end{lem}

We can understand the safety judgement
$\safe{\tid}{\evalf{\lP}{\lint}}{\lcom}{\evalf{\lQ}{\lint}}$ as an obligation
to create a sequence of views $\evalf{\lP}{\lint} = p_1, p_2, \dots, p_{n+1} =
\evalf{\lQ}{\lint}$ for each finite trace $\lpcom_1, \lpcom_2, \dots, \lpcom_n$
of $\lcom$ to justify each transition with action judgements
$\sat{\tid}{\lpcom_1}{p_1}{p_2}$, \dots, $\sat{\tid}{\lpcom_n}{p_n}{p_{n+1}}$.
Thus, when $\safe{\tid}{\evalf{\lP}{\lint}}{\lcom}{\evalf{\lQ}{\lint}}$ holds,
it ensures that every step of the machine correctly preserves a correspondence
between a concrete and abstract execution. Intuitively, the safety judgement
lifts the simulation between concrete and abstract primitive commands
established with action judgements to the implementation and specification of a
method.

In Lemma~\ref{lem:logic2safety}, we establish that the proof judgements of the
logic imply the safety judgements. As a part of the proof, we show that each of
the proof rules of the logic holds of safety judgements. Due to space
constraints, this and other proofs are given in the extended version of the
paper~\cite{ext}.

\FloatBarrier
\newcommand{\lin}{\sqsubseteq}
\newcommand{\hsem}[1]{\mathcal{H}\db{#1}}
\newcommand{\hsemn}[2]{\mathcal{H}_{#1}\db{#2}}

\newcommand{\call}[2]{{\sf call}\ #1(#2)}
\newcommand{\ret}[2]{{\sf ret}\ #1(#2)}

\newcommand{\ccall}[1]{{\sf call}\ #1}
\newcommand{\cret}[2]{{\sf ret}\ #1(#2)}

\newcommand{\idle}{{\sf idle}}

\newcommand{\tidvar}{\tid}

\newcommand{\tp}{\tau}
\newcommand{\atp}{\mathcal{T}}

\newcommand{\preVA}{\mathcal{P}}
\newcommand{\postVA}{\mathcal{Q}}

\newcommand{\hist}{h}

\newcommand{\aval}{a}
\newcommand{\rval}{v}

\section{Soundness}\label{sec:soundness}


In this section, we formulate linearizability for libraries. We also formulate
the soundness theorem, in which we state proof obligations that are necessary
to conclude linearizability.

\mypar{Libraries} We assume a set of method names ${\sf Method}$, ranged over by
$m$, and consider {\em a concrete library} $\llib : {\sf Method} \pto ((\Vals
\times \Vals) \to \LCom)$ that maps method names to commands from $\LCom$, which
are parametrised by a pair of values from $\Vals$. For a given method name $m
\in \dom(\llib)$ and values $\aval, \rval \in \Vals$, a command $\llib(m, \aval,
\rval)$ is an implementation of $m$, which accepts $\aval$ as a method argument
and either returns $\rval$ or does not terminate. Such an unusual way of
specifying method's arguments and return values significantly simplifies further
development, since it does not require modelling a call stack.

Along with the library $\llib$ we consider its specification in the form of an
abstract library $\hlib \in {\sf Method} \pto ((\Vals \times \Vals) \to
\HPCom)$ implementing a set of methods $\dom(\hlib)$ atomically as abstract
primitive commands $\dc{\hlib(m, \aval, \rval) \mid m \in \dom(\hlib)}$
parametrised by an argument $\aval$ and a return value $\rval$. Given a method
$m \in {\sf Method}$, we assume that a parametrised abstract primitive command
$\hlib(m)$ is intended as a specification for $\llib(m)$.

\mypar{Linearizability} The linearizability assumes a complete isolation
between a library and its client, with interactions limited to passing values of
a given data type as parameters or return values of library methods.
Consequently, we are not interested in internal steps recorded in library
computations, but only in the interactions of the library with its client. We
record such interactions using {\em histories}, which are traces including only
events $\call{m}{\aval}$ and $\ret{m}{\rval}$ that indicate an invocation of a
method $m$ with a parameter $\aval$ and returning from $m$ with a return value
$\rval$, or formally:
\[
  \hist::= \emptytr \mid (\tid, \call{m}{\aval}) \cc \hist
  \mid (\tid, \ret{m}{\rval}) \cc\hist.
\]
Given a library $\llib$, we generate all finite histories of $\llib$ by
considering $\NThreads$ threads repeatedly invoking library methods in any
order and with any possible arguments. The execution of methods is described by
semantics of commands from \S~\ref{sec:proglang}.

We define a {\em thread pool} $\tp : \ThreadID \to (\idle \uplus (\LCom \times
\Vals))$ to characterise progress of methods execution in each thread. The case
of $\tp(\tid) = \idle$ corresponds to no method running in a thread $\tid$. When
$\tp(\tid) = (\lcom, \rval)$, to finish some method returning $\rval$ it remains
to execute $\lcom$.
\begin{dfn}\label{def:hsem}
We let $\hsem{\llib, \lstate} = \bigcup_{n \geq 0} \hsemn{n}{\llib, (\lambda
\tid \ldotp \idle), \lstate}$ denote the set of all possible histories of a
library $\llib$ that start from a state $\lstate$, where for a given thread
pool $\tp$, $\hsemn{n}{\llib, \tp, \lstate}$ is defined as a set of histories
such that $\hsemn{0}{\llib, \tp, \lstate} \triangleq \dc{ \emptytr }$ and:
\[
\begin{array}{rcl}
\hsemn{n}{\llib, \tp, \lstate}
& \triangleq &
\{
  ((\tid, \call{m}{\aval}) ::\hist) \mid a \in \Vals \land m \in \dom(\llib) \land
  \tp(\tid) = \idle \land{} \\ & & \hspace{9.25em}
  \exists \rval \ldotp 
  \hist \in \hsemn{n-1}{\llib, \tp[\tid : (\llib(m, \aval, \rval), \rval)], \lstate}
\} \\ & & {} \cup \{
  \hist \mid \exists \tid, \lpcom, \lcom, \lcom', \lstate', \rval \ldotp
  \tp(\tid) = (\lcom, \rval) \land 
  \sttrans{\lcom}{\lstate}{\tid}{\lpcom}{\lcom'}{\lstate'} \land{}
  \\ & & \hspace{9.25em}
  \hist \in \hsemn{n-1}{\llib, \tp[\tid : (\lcom', \rval)], \lstate'}
\} \\ & & {} \cup \{
  ((\tid, \ret{m}{\rval}) ::\hist) \mid m \in \dom(\llib) \land
  \tp(\tid) = (\cskip, \rval) \land{} \\ & & \hspace{9.25em}
  \hist \in \hsemn{n-1}{\llib, \tp[\tid : \idle], \lstate}
\}
\end{array}
\]
\end{dfn}
Thus, we construct the set of all finite histories inductively with all threads
initially idling. At each step of generation, in any idling thread $\tid$ any
method $m \in \dom(\llib)$ may be called with any argument $\aval$ and an
expected return value $\rval$, which leads to adding a command $\llib(m, \aval,
\rval)$ to the thread pool of a thread $\tid$. Also, any thread $\tid$, in which
$\tp(\tid) = (\lcom, \rval)$, may do a transition
$\sttrans{\lcom}{\lstate}{\tid}{\lpcom}{\lcom'}{\lstate'}$ changing a command in
the thread pool and the concrete state. Finally, any thread that has finished
execution of a method's command ($\tp(\tid) = (\cskip, \rval)$) may become idle
by letting $\tp(\tid) = \idle$.

We define $\hsemn{n}{\hlib, \atp, \hstate}$ analogously and let the set of all
histories of an abstract library $\hlib$ starting from the initial state
$\hstate$ be $\hsem{\hlib, \hstate} = \bigcup_{n \geq 0} \hsemn{n}{\hlib,
(\lambda \tid \ldotp \idle), \hstate}$.

\begin{dfn}
  For libraries $\llib$ and $\hlib$ such that $\dom(\llib) = \dom(\hlib)$, we
  say that $\hlib$ \textbfit{linearizes} $\llib$ in the states $\lstate$ and
  $\hstate$, written $(\llib, \lstate) \lin (\hlib, \hstate)$, if $\hsem{\llib,
  \lstate} \subseteq \hsem{\hlib, \hstate}$.  
\end{dfn}
That is, an abstract library $\hlib$ linearizes $\llib$ in the states $\lstate$
and $\hstate$, if every history of $\llib$ can be reproduced by $\hlib$. The
definition is different from the standard one \cite{linearizability}: we use
the result obtained by Gotsman and Yang \cite{linown} stating that the plain
subset inclusion on the sets of histories produced by concrete and abstract
libraries is equivalent to the original definition of linearizability.

\mypar{Soundness w.r.t. linearizability} We now explain proof obligations that
we need to show for every method $m$ of a concrete library $\llib$ to conclude
its linearizability. Particularly, for every thread $\tid$, argument $\aval$,
return value $\rval$, and a command $\llib(m, \aval, \rval)$ we require that
there exist assertions $\preVA(\tid, \hlib(m, \aval, \rval))$ and $\postVA(\tid,
\hlib(m, \aval, \rval))$, for which the following Hoare-style specification
holds:
\begin{equation}\label{eq:soundness1}
    \infer{\tid}{
      \preVA(\tid, \hlib(m, \aval, \rval))
    }{
      \llib(m, \aval, \rval)
    }{
      \postVA(\tid, \hlib(m, \aval, \rval))
    }
\end{equation}
In the specification of $\llib(m, \aval, \rval)$, $\preVA(\tid, \hlib(m, \aval,
\rval))$ and $\postVA(\tid, \hlib(m, \aval, \rval))$ are assertions parametrised
by a thread $\tid$ and an abstract command $\hlib(m, \aval, \rval)$. We require
that in a thread $\tid$ of all states satisfying $\preVA(\tid, \hlib(m, \aval,
\rval))$ and $\postVA(\tid, \hlib(m, \aval, \rval))$ there be only tokens
$\todo{\hlib(m, \aval, \rval)}$ and $\done{\hlib(m, \aval, \rval)}$
respectively:
\begin{multline}\label{eq:soundness2}
\forall \lint, \tid, \lstate, \hstate, \tstate, r \ldotp{} \\[-4pt]
\begin{array}{@{}l@{}}
((\lstate, \hstate, \tstate) \in 
\reif{\evalf{\preVA(\tid, \hlib(m, \aval, \rval))}{\lint} \vop r} \implies
\tstate(\tid) = \todo{\hlib(m, \aval, \rval)})
\\ {} \land
((\lstate, \hstate, \tstate) \in
\reif{\evalf{\postVA(\tid, \hlib(m, \aval, \rval))}{\lint} \vop r} \implies
\tstate(\tid) = \done{\hlib(m, \aval, \rval)})
\end{array}
\end{multline}
Together, (\ref{eq:soundness1}) and (\ref{eq:soundness2}) impose a requirement
that a concrete and an abstract method return the same return value $\rval$. We
also require that the states satisfying the assertions only differ by a token
of a thread $\tid$:
\begin{multline}\label{eq:soundness3}
\forall \lint, \tid, \hpcom, \hpcom', r, \tstate \ldotp 
  (\lstate, \hstate, \tstate [\tid : \done{\hpcom}]) \in 
    \reif{\evalf{\postVA(\tid, \hpcom)}{\lint} \vop r} \iff{} \\
  (\lstate, \hstate, \tstate [\tid : \todo{\hpcom'}]) \in 
    \reif{\evalf{\preVA(\tid, \hpcom')}{\lint} \vop r}.
\end{multline}

\begin{thm}\label{thm:lin}
For given libraries $\llib$ and $\hlib$ together with states $\lstate$ and
$\hstate$, $(\llib, \lstate) \lin (\hlib, \hstate)$ holds, if $\dom(\llib) =
\dom(\hlib)$ and (\ref{eq:soundness1}), (\ref{eq:soundness2}) and
(\ref{eq:soundness3}) hold for every method $m$, thread $\tid$ and values
$\aval$ and $\rval$.
\end{thm}

\FloatBarrier

\newcommand{\rgassert}[2]{(#1) \models #2}
\newcommand{\rginfer}[5]{#2 \infer{#1}{#3}{#4}{#5}}
\newcommand{\pRG}{P}
\newcommand{\qRG}{Q}
\newcommand{\rely}{R}
\newcommand{\guar}{G}
\newcommand{\lRely}{\mathcal{R}}
\newcommand{\lGuar}{\mathcal{G}}
\newcommand{\lPRG}{\pi}
\newcommand{\lAct}{\mathcal{A}}

\newcommand{\LStateH}{\LState_{\sf H}}
\newcommand{\HStateH}{\HState_{\sf H}}

\newcommand{\locs}{l}
\newcommand{\alocs}{L}
\newcommand{\loc}{c}
\newcommand{\aloc}{c}
\newcommand{\tids}{T}

\newcommand{\rgact}[2]{#1 \leadsto #2}

\newcommand{\vopRG}{\mathbin{\vop{}}}

\section{The RGSep-based Logic\label{sec:rgsep}}


In this section, we demonstrate an instance of the generic proof system that is
capable of handling algorithms with helping. This instance is based on RGSep
\cite{rgsep}, which combines rely-guarantee reasoning \cite{rg} with separation
logic \cite{seplogic}.

The main idea of the logic is to partition the state into several thread-local
parts (which can only be accessed by corresponding threads) and the shared part
(which can be accessed by all threads). The partitioning is defined by proofs
in the logic: an assertion in the code of a thread restricts its local state
and the shared state. In addition, the partitioning is dynamic, meaning that
resources, such as a part of a heap or a token, can be moved from the local
state of a thread into the shared state and vice versa. By transferring a token
to the shared state, a thread gives to its environment a permission to change
the abstract state. This allows us to reason about environment helping that
thread.

\mypar{The RGSep-based view monoid} Similarly to DCSL, we assume that states
represent heaps, i.e. that $\LState = \HState = \Loc \ptofin \Vals \uplus
\dc{\fault}$, and we denote all states but a faulting one with $\LStateH =
\HStateH = \Loc \ptofin \Vals$. We also assume a standard set of
heap-manipulating primitive commands with usual semantics.

We define views as triples consisting of three components: a predicate $\pRG$
and binary relations $\rely$ and $\guar$. A predicate $\pRG \in
\powerset{(\LStateH \times \HStateH \times \Tokens)^2}$ is a set of pairs $(l,
s)$ of local and shared parts of the state, where each part consists of
concrete state, abstract state and tokens. {\em Guarantee} $\guar$ and {\em
rely} $\rely$ are relations from $\powerset{(\LState \times \HState \times
\Tokens)^2}$, which summarise how individual primitive commands executed by the
method's thread (in case of $\guar$) and the environment (in case of $\rely$)
may change the shared state. Together guarantee and rely establish a protocol
that views of the method and its environment respectively must agree on each
other's transitions, which allows us to reason about every thread separately
without considering local state of other threads, assuming that they follow the
protocol. The agreement is expressed with the help of a well-formedness
condition on views of the RGSep-based monoid that their predicates must be {\em
stable} under rely, meaning that their predicates take into account whatever
changes their environment can make:
\[
\stable(\pRG, \rely) \triangleq \forall l, s, s' \ldotp
(l, s) \in \pRG \land
(s, s') \in \rely \implies
(l, s') \in \pRG.
\]
\noindent
A predicate that is stable under rely cannot be invalidated by any state
transition from rely. Stable predicates with rely and guarantee relations form
the view monoid with the underlying set of views 
$
\Views_\RGsep = \{ (\pRG, \rely, \guar) \mid \stable(\pRG, \rely) \} \cup \{ \bot \},
$
where $\bot$ denotes a special {\em inconsistent} view with the empty
reification.
The reification of other views simply joins shared and local parts of the
state:
\[
\reif{(\pRG, \rely, \guar)} = \{ (\lstate_l \slop \lstate_s, \hstate_l \slop
\hstate_s, \tstate_l \uplus \tstate_s) \mid ((\lstate_l, \hstate_l, \tstate_l),
(\lstate_s, \hstate_s, \tstate_s))
\in \pRG \}.
\]

Let an operation $\cdot$ be defined on states analogously to DCSL. Given
predicates $\pRG$ and $\pRG'$, we let $\pRG \vop \pRG'$ be a predicate denoting
the pairs of local and shared states in which the local state can be divided
into two substates such that one of them together with the shared state
satisfies $\pRG$ and the other together with the shared state satisfies
$\pRG'$:
\[
\pRG \vop \pRG' \triangleq \{ ((\lstate_l \slop \lstate'_l, \hstate_l \slop
\hstate'_l, \tstate_l \uplus \tstate'_l), s) \mid ((\lstate_l, \hstate_l,
\tstate_l), s) \in \pRG \land ((\lstate'_l, \hstate'_l, \tstate'_l), s) \in
\pRG' \}
\]
We now define the monoid operation $\vop$, which we use to compose views of
different threads. When composing views $(\pRG, \rely, \guar)$ and $(\pRG',
\rely', \guar')$ of the parallel threads, we require predicates of both
to be immune to interference by all other threads and each other. Otherwise,
the result is inconsistent:
\[
(\pRG, \rely, \guar) \vop (\pRG', \rely', \guar') \triangleq
  \mbox{if } \guar \subseteq  \rely' \land \guar' \subseteq \rely
  \mbox{ then }
  (\pRG \vop \pRG', \rely \cap \rely', \guar \cup \guar')
  \mbox{ else } \bot.
\]
That is, we let the composition of views be consistently defined when the state
transitions allowed in a guarantee of one thread are treated as environment
transitions in the other thread, i.e. $\guar \subseteq \rely'$ and $\guar'
\subseteq \rely$. The rely of the composition is $\rely \cap \rely'$, since the
predicate $\pRG * \pRG'$ is guaranteed to be stable only under environment
transitions described by both $\rely$ and $\rely'$. The guarantee of the composition is $\guar \cup \guar'$, since other views need to take into account all state transitions either from $\guar$ or from $\guar'$.

\begin{figure}[t!]
$\vDash: (\LState \times \HState \times \Tokens) \times (\LState \times \HState \times \Tokens) \times \Int \times \Assn$\\[3pt]
\hfill $
\begin{array}{l@{\ \ }l}
\rgassert{(\lstate_l, \hstate_l, \tstate_l), (\lstate_s, \hstate_s, \tstate_s),
  \lint}{E \slto F},
  & \mbox{iff } \lstate_l = [\db{E}_\lint : \db{F}_\lint], \hstate_l =
    \emptymap, \mbox{ and } \tstate_l = \emptymap\\[1pt]
\rgassert{(\lstate_l, \hstate_l, \tstate_l), (\lstate_s, \hstate_s, \tstate_s), 
  \lint}{E \aslto F},
  & \mbox{iff } \lstate_l = \emptymap, 
    \hstate_l = [\db{E}_\lint : \db{F}_\lint], \mbox{ and }  \tstate_l =
    \emptymap \\[1pt]
\rgassert{(\lstate_l, \hstate_l, \tstate_l), (\lstate_s, \hstate_s, \tstate_s),
  \lint}{\ltodo{\tid}{\hpcom}},
  & \mbox{iff } \lstate_l = \emptymap, \hstate_l = \emptymap, \mbox{ and } 
    \tstate_l = [\tid : \todo{\hpcom}]\\[1pt]
\rgassert{(\lstate_l, \hstate_l, \tstate_l), (\lstate_s, \hstate_s, \tstate_s), 
  \lint}{\ldone{\tid}{\hpcom}},
  & \mbox{iff } \lstate_l = \emptymap, \hstate_l = \emptymap, \mbox{ and } 
    \tstate_l = [\tid : \done{\hpcom}]\\[1pt]
\rgassert{(\lstate_l, \hstate_l, \tstate_l), (\lstate_s, \hstate_s, \tstate_s),
  \lint}{\boxed{\pi}},
  & \mbox{iff } \lstate_l = \emptymap, \hstate_l = \emptymap, \tstate_l =
    \emptymap, \mbox{ and} \\ & \hfill
    \rgassert{(\lstate_s, \hstate_s, \tstate_s), (\emptymap, \emptymap,
    \emptymap), \lint}{\lPRG}\\[1pt]
\rgassert{(\lstate_l, \hstate_l, \tstate_l), (\lstate_s, \hstate_s, \tstate_s),
  \lint}{\pi * \pi'},
  & \mbox{iff there exist } \lstate'_l, \lstate''_l, \hstate'_l, \hstate''_l, \tstate'_l, \tstate''_l \mbox{ such that}\\ &
  \lstate_l = \lstate'_l \slop \lstate''_l, 
    \hstate_l = \hstate'_l \slop \hstate''_l,
    \tstate_l = \tstate'_l \uplus \tstate''_l, \\ &
    \rgassert{(\lstate'_l, \hstate'_l, \tstate'_l), (\lstate_s, \hstate_s, \tstate_s),
  \lint}{\pi}, \mbox{ and }\\ &
    \rgassert{(\lstate''_l, \hstate''_l, \tstate''_l), (\lstate_s, \hstate_s, \tstate_s),
  \lint}{\pi'}
\end{array}
$
\caption{Satisfaction relation for a fragment of the assertion language ${\sf VAssn}$\label{fig:RGsem4frm}}
\end{figure}

\mypar{The RGSep-based program logic}
We define the view assertion language $\VAssn$ that is a parameter of the proof
system. Each view assertion $\vassn$ takes form of a triple $(\lPRG,
\lRely, \lGuar)$, and the syntax for $\lPRG$ is:
\[
\begin{array}{c c l}
  E    & ::= & a \mid X \mid E+E \mid \dots, \quad \mbox{where } X \in \LVars, a \in \Vals \\
  \lPRG    & ::= & E = E \mid E \slto E \mid E \aslto E \mid \ltodo{\tid}{\hpcom}
              \mid \ldone{\tid}{\hpcom} \mid \boxed{\lPRG} \mid
                   \lPRG \vop \lPRG \mid \neg \lPRG \mid \dots \\
\end{array}
\]
Formula $\lPRG$ denotes a predicate of a view as defined by a satisfaction
relation $\models$ in Figure~\ref{fig:RGsem4frm}. There $E \slto E$ and $E
\aslto E$ denote a concrete and an abstract state describing singleton heaps. A
non-boxed formula $\lPRG$ denotes the view with the local state satisfying
$\lPRG$ and shared state unrestricted; $\boxed{\lPRG}$ denotes the view with the
empty local state and the shared state satisfying $\lPRG$; $\lPRG \vop \lPRG'$
the composition of predicates corresponding to $\lPRG$ and $\lPRG'$. The
semantics of the rest of connectives is standard. Additionally, for simplicity
of presentation of the syntax, we require that boxed assertions $\boxed{\lPRG}$
be not nested (as opposed to preventing that in the definition).

The other components $\lRely$ and $\lGuar$ of a view assertion are sets of {\em
rely/guarantee actions} $\lAct$ with the syntax: $\lAct ::=
\rgact{\lPRG}{\lPRG'}$. An action $\rgact{\lPRG}{\lPRG'}$ denotes a change of a
part of the shared state that satisfies $\lPRG$ into one that satisfies
$\lPRG'$, while leaving the rest of the shared state unchanged. We associate
with an action $\rgact{\lPRG}{\lPRG'}$ all state transitions from the following
set:
\begin{multline*}
\db{\rgact{\lPRG}{\lPRG'}} = 
\{ ((\lstate_s \slop \lstate''_s, \hstate_s \slop \hstate''_s, \tstate_s \uplus
\tstate''_s), 
(\lstate'_s \slop \lstate''_s, \hstate'_s \slop \hstate''_s, \tstate'_s \uplus
\tstate''_s))
\mid{} \\ \exists \lint \ldotp
\rgassert{(\emptymap, \emptymap, \emptymap), (\lstate_s, \hstate_s, \tstate_s), \lint}{\boxed{\lPRG}} \land
\rgassert{(\emptymap, \emptymap, \emptymap), (\lstate'_s, \hstate'_s, \tstate'_s), \lint}{\boxed{\lPRG'}} \}
\end{multline*}

We give semantics to view assertions with the function
$\evalf{\cdot}{\cdot}$ that is defined as follows:
\[
\evalf{(\lPRG, \lRely, \lGuar)}{\lint} \triangleq 
  (\dc{ (l, s) \mid 
    \rgassert{l, s, \lint}{\lPRG}}, \bigcup\nolimits_{\lAct \in \lRely} \db{\lAct}, \bigcup\nolimits_{\lAct \in \lGuar} \db{\lAct}).
\]

\FloatBarrier

\newcommand{\shared}{{\sf global}}

\newcommand{\cnt}{{\tt k}}
\newcommand{\acnt}{{\tt K}}
\newcommand{\comb}{{\tt L}}

\newcommand{\mytid}{{\tt mytid}}
\newcommand{\tinv}{{\sf tinv}}

\newcommand{\lparam}{{\tt arg}}
\newcommand{\lres}{{\tt res}}

\newcommand{\minc}{{\sf inc}}
\newcommand{\tidp}{\tid'}

\newcommand{\nulltid}{0}

\newcommand{\addr}[1]{\&{#1}}

\newcommand{\tpre}[3]{{\sf todo}({#1}, {#2}, {#3})}
\newcommand{\tpost}[3]{{\sf done}({#1}, {#2}, {#3})}

\newcommand{\nil}{{\tt nil}}

\newcommand{\inv}{{\sf kinv}}

\newcommand{\taskt}{{\sf task_{todo}}}
\newcommand{\taskd}{{\sf task_{done}}}
\newcommand{\alltinv}{\vopall_{j \in \ThreadID}\,\tinv(j)}
\newcommand{\nslto}{\mathbin{\not\mapsto}}
\newcommand{\LI}{{\sf LI}}

\section{Example}


{

\renewcommand\floatpagefraction{.9}
\renewcommand\dblfloatpagefraction{.9} 
\renewcommand\topfraction{.9}
\renewcommand\dbltopfraction{.9} 
\renewcommand\bottomfraction{.9}
\renewcommand\textfraction{.1}   
\setcounter{totalnumber}{50}
\setcounter{topnumber}{50}
\setcounter{bottomnumber}{50}
\setlength{\textfloatsep}{5pt plus 1pt minus 1pt}
\setlength{\floatsep}{5pt plus 1pt minus 1pt}

\begin{figure}[t]
\begin{lstlisting}[basicstyle=\small\ttfamily]
int L = 0, #$\cnt$# = 0, #$\lparam$#[N], #$\lres$#[N]; \\ initially all #$\lres[i] \neq \nil$#

#$\llib(\minc, a, r)$#:
  #\lstassert{
    \shared * M(\tid) * \ltodo{\tid}{\hlib(\minc, a, r)}
    }#
  #$\lparam$#[mytid()] := #$a$#;
  #$\lres$#[mytid()] := #$\nil$#; $\label{line:task}$
  #\lstassert{
    \shared *
    \boxed{
      {\sf true} * (\taskt(\tid, a, r) \lor \taskd(\tid, a, r))
    }
    }#
  while (#$\lres$#[mytid()] = #$\nil$#): $\label{line:loop}$
    if (CAS(&L, #$\nulltid$#, mytid())):
      #\lstassert{
    \boxed{
      \addr{\comb} \slto \tid * \alltinv} *
    \boxed{
      {\sf true} * (\taskt(\tid, a, r) \lor \taskd(\tid, a, r))
    } * \inv(\_)
    }#
      for (i := 1; i #$\leq$# N; ++i):
        #\lstassert{
    \boxed{
      \addr{\comb} \slto \tid * \alltinv
    } * \inv(\_) * \LI(i, \tid, a, r)
    }#
        if (#$\lres$#[i] = #$\nil$#):
          #\lstassert{
    \exists V, A, R \sothat
    \inv(V) * \LI(i, \tid, a, r)
    *{} \\
    \boxed{
      \addr{\comb} \slto \tid * \alltinv
    }* \boxed{ 
    {\sf true} * \taskt(i, A, R)
    }
    }#
          #$\cnt$# := #$\cnt$# + #$\lparam$#[i];
          #\lstassert{
    \exists V, A, R \sothat
    \addr{\cnt} \slto V+A * \addr{\acnt} \aslto V * \LI(i, \tid, a, r)
    *{}\\
    \boxed{
      \addr{\comb} \slto \tid * \alltinv
    } * \boxed{ 
    {\sf true} * \taskt(i, A, R)
    }
    }#
          #$\lres$#[i]  := #$\cnt$#; $\label{line:lp}$
          #\lstassert{
    \exists V, A, R \sothat
    \inv(V+A) * \LI(i+1, \tid, a, r)
     *{}\\ \boxed{
      \addr{\comb} \slto \tid * \alltinv
    } * \boxed{ 
    {\sf true} * \taskd(i, A, R)
    }
    }#
      #\lstassert{
    \boxed{
      \addr{\comb} \slto \tid * \alltinv
    } * 
    \boxed{
      {\sf true} * \taskd(\tid, a, r)
      } * \inv(\_)
    }#
      L = #$\nulltid$#;
  assume(#$\lres$#[mytid()] = #$r$#);
  #\lstassert{
  \shared * M(\tid) * \ldone{\tid}{\hlib(\minc, a, r)}
  }#
\end{lstlisting}
\vspace{-1.2em}
\caption{Proof outline for a flat combiner of a concurrent increment.
Indentation is used for grouping commands.}
\label{fig:fcproof}
\end{figure}

}


In this section, we demonstrate how to reason about algorithms with helping
using relational views. We choose a simple library $\llib$ implementing a
concurrent increment and prove its linearizability with the RGSep-based logic.

The concrete library $\llib$ has one method $\minc$, which increments the value
of a shared counter $\cnt$ by the argument of the method. The specification of
$\llib$ is given by an abstract library $\hlib$. The abstract command, provided
by $\hlib$ as an implementation of $\minc$, operates with an abstract counter
$\acnt$ as follows (assuming that $\acnt$ is initialised by zero):
\begin{lstlisting}[basicstyle=\small\ttfamily]
#$\hlib(\minc, a, r)$#:   < __kabs := __kabs + #$a$#; assume(__kabs == #$r$#); >
\end{lstlisting}
That is, $\hlib(\minc, a, r)$ atomically increments a counter and a command
${\tt assume}(\acnt == r)$, which terminates only if the return value $r$ chosen
at the invocation equals to the resulting value of $\acnt$. This corresponds to
how we specify methods' return values in
\S\ref{sec:soundness}.

In Figure~\ref{fig:fcproof}, we show the pseudo-code of the implementation of a
method $\minc$ in a C-style language along with a proof outline. The method
$\llib(\minc, a, r)$ takes one argument, increments a shared counter $\cnt$ by
it and returns the increased value of the counter. Since $\cnt$ is shared among
threads, they follow a protocol regulating the access to the counter. This
protocol is based on flat combining \cite{flatcombining}, which is a
synchronisation technique enabling a parallel execution of sequential
operations.

The protocol is the following. When a thread $\tid$ executes $\llib(\minc, a,
r)$, it first makes the argument of the method visible to other threads by
storing it in an array $\lparam$, and lets $\lres[\tid] = \nil$ to signal to
other threads its intention to execute an increment with that argument. It then
spins in the loop on line~\ref{line:loop}, trying to write its thread identifier
into a variable $\comb$ with a compare-and-swap (CAS). Out of all threads
spinning in the loop, the one that succeeds in writing into $\comb$ becomes a
{\em combiner}: it performs the increments requested by all threads with
arguments stored in $\lparam$ and writes the results into corresponding cells of
the array $\lres$. The other threads keep spinning and periodically checking the
value of their cells in $\lres$ until a non-$\nil$ value appears in it, meaning
that a combiner has performed the operation requested and marked it as finished.
The protocol relies on the assumption that $\nil$ is a value that is never
returned by the method. Similarly to the specification of the increment method,
the implementation in Figure~\ref{fig:fcproof} ends with a command ${\tt
assume}({\tt res}[{\tt mytid}()] = r)$.

The proof outline features auxiliary assertions defined in
Figure~\ref{fig:preds}. In the assertions we let $\_$ denote a value or a
logical variable whose name is irrelevant. We assume that each program variable
${\tt var}$ has a unique location in the heap and denote it with $\& {\tt var}$.
Values $a$, $r$ and $\tid$ are used in the formulas and the code as constants.

We prove the following specification for $\llib(\minc, a, r)$:
\[
\rginfer{\tid}{\lRely_\tid, \lGuar_\tid}{
  \nlfrac{
    \shared * M(\tid) *
  }{
    \ltodo{\tid}{\hlib(\minc, a, r)}
  }
}{
  \llib(\minc, a, r)
}{
  \nlfrac{
    \shared * M(\tid) *
  }{
    \ldone{\tid}{\hlib(\minc, a, r)}
  }
}
\]
In the specification, $M(\tid)$ asserts the presence of $\lparam[\tid]$ and
$\lres[\tid]$ in the shared state, and $\shared$ is an assertion describing the
shared state of all the threads. Thus, the pre- and postcondition of the
specification differ only by the kind of token given to $\tid$.

\begin{figure}[t]
$
\begin{array}{l c l}
X \nslto Y & \triangleq & \exists Y' \sothat X \slto Y' * Y \neq Y'
\\[1pt]
M(\tid) & \triangleq & \boxed{
{\sf true} * (\addr{\lparam[\tid]} \slto \_ * \addr{\lres[\tid]} \nslto \nil)} \\[1pt]
\taskt(\tid, a, r) & \triangleq & \addr{\lparam[\tid]} \slto a * \addr{\lres[\tid]} \slto \nil * \ltodo{\tid}{\hlib(\minc, a, r)};
\\[1pt]
\taskd(\tid, a, r) & \triangleq & \addr{\lparam[\tid]} \slto a * \addr{\lres[\tid]} \slto r * r \not= \nil * \ldone{\tid}{\hlib(\minc, a, r)};
\\[1pt]
\inv(V) & \triangleq & \addr{\cnt} \slto V * \addr{\acnt} \aslto V
\\[1pt]
\LI(i, \tid, a, r) & \triangleq & \boxed{
\begin{array}[t]{l}
{\sf true} * ((\tid < i \land
  \taskd(\tid, a, r))\lor{} \\[1pt]
  (\tid \geq i \land (\taskt(\tid, a, r) \lor \taskd(\tid, a, r))))
\end{array}}
\\[1pt]
\tinv(i) & \triangleq & \addr{\lparam[i]} \slto \_ * \addr{\lres[i]} \aslto \_ \lor \taskt(i, \_, \_) \lor \taskd(i, \_, \_) \\[1pt]
\shared & \triangleq & \boxed{
    (\addr{\comb} \slto 0 * \inv(\_) \lor \addr{\comb} \nslto 0) * \alltinv},
\end{array}
$

\caption{Auxiliary predicates. $\vopall_{j \in \ThreadID}\,\tinv(j)$ denotes
$\tinv(1) \vop \tinv(2) \vop \cdot \vop \tinv(\NThreads)$}
\label{fig:preds}
\end{figure}
\FloatBarrier

The main idea of the proof is in allowing a thread $\tid$ to share the ownership
of its token $\ltodo{\tid}{\hlib(\minc, a, r)}$ with the other threads. This
enables two possibilities for $\tid$. Firstly, $\tid$ may become a combiner.
Then $\tid$ has a linearization point on line~\ref{line:lp} (when the loop index
$i$ equals to $\tid$). In this case $\tid$ also {\em helps} other
concurrent threads by performing their linearization points on
line~\ref{line:lp} (when $i \neq \tid$). The alternative possibility is that
some other thread becomes a combiner and does a linearization point of $\tid$.
Thus, the method has a non-fixed linearization point, as it may occur in the
code of a different thread.

We further explain how the tokens are transferred. On line~\ref{line:task} the
method performs the assignment \texttt{res[mytid()] := nil}, signalling to other
threads about a task this thread is performing. At this step, the method
transfers its token $\ltodo{\tid}{\hlib(\minc, a, r)}$ to the shared state, as
represented by the assertion $\boxed{{\sf true} * \taskt(\tid, a, r)}$. In order
to take into consideration other threads interfering with $\tid$ and possibly
helping it, here and further we stabilise the assertion by adding a disjunct
$\taskd(\tid, a, r)$.

If a thread $\tid$ gets help from other threads, then $\taskd(\tid, a, r)$
holds, which implies that $\lres[\tid] \neq \nil$ and $\tid$ cannot enter the
loop on line~\ref{line:loop}. Otherwise, if $\tid$ becomes a combiner, it
transfers $\inv(\_)$ from the shared state to the local state of $\tid$ to take
over the ownership of the counters $\cnt$ and $\acnt$ and thus ensure that the
access to the counter is governed by the mutual exclusion protocol. At each
iteration $i$ of the forall loop, \texttt{res[i] = nil} implies that $\taskt(i,
\_, \_)$ holds, meaning that there is a token of a thread $i$ in the shared
state. Consequently, on line~\ref{line:lp} a thread $\tid$ may use it to perform
a linearization point of $i$.

The actions defining the guarantee relation $\lGuar_\tid$ of a thread $\tidp$
are the following:
\begin{enumerate}
\item $\rgact{\addr{\lparam[\tid]} \slto \_ * \addr{\lres[\tid]} \nslto
\nil}{\addr{\lparam[\tid]} \slto a * \addr{\lres[\tid]} \nslto \nil}$;
\item $\rgact{\addr{\lparam[\tid]} \slto a * \addr{\lres[\tid]} \nslto
\nil}{\taskt(\tid, a, r)}$;
\item $\rgact{\addr{\comb} \slto 0 * \inv(\_)}{\addr{\comb} \slto
    \tid}$;
\item $\rgact{\addr{\comb} \slto \tid * \taskt(T, A, R)}{\addr{\comb} \slto
\tid * \taskd(T, A, R)}$
\item $\rgact{\addr{\comb} \slto \tid}{\addr{\comb} \slto 0 * \inv(\_)}$
\item $\rgact{\taskd(\tid, a, r)}{\addr{\lparam[\tid]} \slto a *
\addr{\lres[\tid]} \slto r}$
\end{enumerate}
Out of them, conditions 2 and 6 specify transfering the token of a thread
$\tid$ to and from the shared state, and condition 4 describes using the shared
token of a thread $T$. The rely relation of a thread $\tid$ is then defined as
the union of all actions from guarantee relations of other threads and an
additional action for each thread $\tid' \in \ThreadID \setminus \{ \tid \}$
allowing the client to prepare a thread $\tid'$ for a new method call by giving
it a new token: $\rgact{\ldone{\tid'}{\hlib(\minc, A,
R)}}{\ltodo{\tid'}{\hlib(\minc, A', R')}}$.

\section{Related Work}
\label{sec:related}

There has been a significant amount of research on methods for proving
linearizability. Due to space constraints, we do not attempt a comprehensive
survey here (see~\cite{DongolD14}) and only describe the most closely related
work.

The existing logics for linearizability that use linearization points differ in
the thread-modular reasoning method used and, hence, in the range of concurrent
algorithms that they can handle. Our goal in this paper was to propose a uniform
basis for designing such logics and to formalise the method they use for
reasoning about linearizability in a way independent of the particular
thread-modular reasoning method used. We have only shown instantiations of our
logic based on disjoint concurrent separation logic~\cite{csl} and
RGSep~\cite{rgsep}. However, we expect that our logic can also be instantiated
with more complex thread-modular reasoning methods, such as those based on
concurrent abstract predicates~\cite{CAP} or islands and
protocols~\cite{turon-popl13}.

Our notion of tokens is based on the idea of treating method specifications as
resources when proving atomicity, which has appeared in various guises in
several logics~\cite{rgsep,rgsim-lin,tada}. Our contribution is to formalise
this method of handling linearization points independently from the underlying
thread-modular reasoning method and to formulate the conditions for soundly
combining the two (Definition~\ref{def:axiom}, \S\ref{sec:logic}).

We have presented a logic that unifies the various logics based on linearization
points with helping. However, much work still remains as this reasoning method
cannot handle all algorithms. Some logics have introduced {\em speculative}
linearization points to increase their
applicability~\cite{turon-popl13,rgsim-lin}; our approach to helping is closely
related to this, and we hope could be extended to speculation. But there are
still examples beyond this form of reasoning: for instance there are no proofs
of the Herlihy-Wing queue~\cite{linearizability} using linearization points
(with helping and/or speculation). This algorithm can be shown linearizable
using forwards/backwards simulation~\cite{linearizability} and more recently has
been shown to only require a backwards simulation~\cite{schellhorn-cav12}. But
integrating this form of simulation with the more intrincate notions of
interference expressible in the Views framework remains an open problem.

Another approach to proving linearizability is the aspect-oriented method. This
gives a series of properties of a queue~\cite{henzinger-concur13} (or a
stack~\cite{dodds-popl15}) implementation which imply that the implementation is
linearizable. This method been applied to algorithms that cannot be handled with
standard linearization-point-based methods. However, the aspect-oriented
approach requires a custom theorem per data structure, which limits its
applicability.

In this paper we concentrated on linearizability in its original
form~\cite{linearizability}, which considers only finite computations and,
hence, specifies only safety properties of the library. Linearizability has
since been generalised to also specify liveness
properties~\cite{gotsman-icalp11}. Another direction of future work is to
generalise our logic to handle liveness, possibly building on ideas
from~\cite{liang-lics14}.

When a library is linearizable, one can use its atomic specification instead of
the actual implementation to reason about its clients~\cite{filipovic-tcs}. Some
logics achieve the same effect without using linearizability, by expressing
library specifications as judgements in the logic rather than as the code of an
abstract library~\cite{iris,icap,sergey-esop15}. It is an interesting direction
of future work to determine a precise relationship between this method of
specification and linearizability, and to propose a generic logic unifying the
two.

\section{Conclusion}

We have presented a logic for proving the linearizability of concurrent
libraries that can be instantiated with different methods for thread-modular
reasoning. To this end, we have extended the Views framework~\cite{views} to
reason about relations between programs. Our main technical contribution in
this regard was to propose the requirement for axiom soundness
(Definition~\ref{def:axiom}, \S\ref{sec:logic}) that ensures a correct
interaction between the treatment of linearization points and the underlying
thread-modular reasoning.  We have shown that our logic is powerful enough to
handle concurrent algorithms with challenging features, such as helping. More
generally, our work marks the first step towards unifying the logics for
proving relational properties of concurrent programs.


\bibliographystyle{abbrv}
\bibliography{main}

\iflong
\clearpage
\appendix

\newcommand{\uminus}{\setminus}

\newcommand{\tpinv}{{\sf inv}}
\newcommand{\tpinvall}[5]{(\forall \tidk \sothat \tpinv_\tidk(#1, #2, #3, #4, #5))}
\newcommand{\mip}{{\sf MiP}}
\newcommand{\vP}[2]{p}
\newcommand{\vQ}[2]{q}
\newcommand{\vV}[1]{v_{#1}}

\newcommand{\safelib}{{\sf safelib}}
\newcommand{\gfp}[1]{{\sf gfp}\,{#1}}
\newcommand{\fr}[2]{{\sf fr}_{#2}({#1})}
\newcommand{\fun}[1]{\phi(#1)}
\newcommand{\funx}[1]{\psi(#1)}
\newcommand{\fbin}[2]{\phi(#1, #2)}

\newcommand{\sfsafe}{{\sf safe}_\tid}

\newcommand{\tidk}{k}
\newcommand{\ag}[1]{}

\section{Additional details of the RGSep-based logic}

The unit $\unit_\RGsep$ does not restrict states and the allowed state
transitions of the environment, while disallowing any action in the current
thread:
\[
\unit_\RGsep = (
\{ (\emptymap, \emptymap, \emptymap) \} \times (\LState \times \HState \times \Tokens), (\LState \times \HState \times \Tokens)^2, \emptyset
)
\]

Since action judgements are essential for reasoning about primitive commands in
our logic, we further refine conditions under which it holds of views from the
RGSep-based view monoid.
\begin{prop}\label{prop:logic-rgaxioms}
  The action judgement
  $\sat{\tid}{\lpcom}{(\pRG, \rely, \guar)}{(\qRG, \rely, \guar)}$ holds, if it
  is true that:
\begin{itemize}
\item
$\begin{multlined}[t]
\forall \lstate_l, \lstate_s, \lstate'_l, \lstate'_s, \hstate_l, \hstate_s,
\tstate_l, \tstate_s \ldotp
((\lstate_l,\hstate_l,\tstate_l), (\lstate_s,\hstate_s,\tstate_s)) \in \pRG
\land{} \\
\lstate'_l \slop \lstate'_s \in \intp{\tid}{\lpcom}{\lstate_l \slop \lstate_s} 
\implies
\exists \hstate'_l, \hstate'_s, \tstate'_l, \tstate'_s \ldotp
((\lstate'_l,\hstate'_l,\tstate'_l), (\lstate'_s,\hstate'_s,\tstate'_s))\in\qRG
\land{}\\
((\lstate_s,\hstate_s,\tstate_s),(\lstate'_s,\hstate'_s,\tstate'_s)) \in \guar
\land \lptrans{
  \hstate_l \slop \hstate_s, \tstate_l \uplus \tstate_s
}{
  \hstate'_l \slop \hstate'_s, \tstate'_l \uplus \tstate'_s
};
\end{multlined}$
\item
$\intp{\tid}{\lpcom}{\lstate} \neq \fault \implies
    \forall \lstate' \ldotp
    \intp{\tid}{\lpcom}{\lstate \slop \lstate'} = \{
      \lstate'' \slop \lstate' \mid \lstate'' \in \intp{\tid}{\lpcom}{\lstate}
    \}$.
\end{itemize}
\end{prop}
The requirement to primitive commands in Proposition~\ref{prop:logic-rgaxioms}
is similar to that of the action judgements. The difference is that in the
RG-based proof system it is not necessary to require $\lpcom$ to preserve any
view $r$ of the environment: since a predicate $\pRG_r$ of any view $(\pRG_r,
\rely_r, \guar_r)$ in another thread is stable under $\rely_r$, it is also
stable under $\guar \subseteq \rely_r$ whenever $(\pRG, \rely, \guar) \vop
(\pRG_r, \rely_r, \guar_r)$ is defined. Consequently, views of the environment
are never invalidated by local transitions. Using the premise of
Proposition~\ref{prop:logic-rgaxioms} in {\sc Prim} rule makes it closer to the
standard proof rule for the atomic step in Rely/Guarantee.

\section{Compositionality properties of the safety relation}

In this section, we formulate and prove compositionality properties of the
safety relation.

For further reference we restate the definition of a repartitioning
implication:
\begin{equation}\label{def:rimpl}
p \rimpl q \triangleq \forall r \sothat \reif{p * r} \subseteq \reif{q * r}.
\end{equation}

\begin{lem}\label{lem:safe}
The safety relation $\sfsafe$ has the following closure properties:
\begin{itemize}
\item {\sc Frame}: $\forall \tid, \lcom, p, q, r \sothat 
  \safe{\tid}{p}{\lcom}{q} \implies
  \safe{\tid}{p * r}{\lcom}{q * r}$;
\item {\sc Choice}: $\forall \tid, \lcom_1, \lcom_2, p, q \sothat 
  \safe{\tid}{p}{\lcom_1}{q} \land \safe{\tid}{p}{\lcom_2}{q} \implies
  \safe{\tid}{p}{\lcom_1 \ldet \lcom_2}{q}$;
\item {\sc Iter}: $\forall \tid, \lcom, p \sothat 
  \safe{\tid}{p}{\lcom}{p} \implies \safe{\tid}{p}{\iter{\lcom}}{p}$;
\item {\sc Seq}: $\forall \tid, \lcom_1, \lcom_2, p, p', q \sothat 
  \safe{\tid}{p}{\lcom_1}{p'} \land \safe{\tid}{p'}{\lcom_2}{q} \implies
  \safe{\tid}{p}{\lcom_1 \lseq \lcom_2}{q}$;
\item {\sc Conseq}: $\forall \tid, \lcom, p, p', q, q' \sothat p' \rimpl p
  \land \safe{\tid}{p}{\lcom}{q} \land q \rimpl q' \implies
  \safe{\tid}{p'}{\lcom}{q'}$;
\item {\sc Disj}: $\begin{multlined}[t]
  \forall \tid, \lcom, p_1, p_2, q_1, q_2 \sothat 
  \safe{\tid}{p_1}{\lcom}{q_1} \land \safe{\tid}{p_2}{\lcom}{q_2} \implies{}\\
  \safe{\tid}{p_1 \lor p_2}{\lcom}{q_1 \lor q_2}.\end{multlined}$
\end{itemize}
\end{lem}

We prove all of the properties by coinduction. To this end, we take
the fixed-point 
definition of $\sfsafe$. We consider 
\[
F_\tid : \powerset{\Views \times \LCom \times \Views} \to
                                  \powerset{\Views \times \LCom \times \Views}
\]
defined as follows:
\[
\begin{array}{rcl}
F_\tid(X) & \triangleq &
\dc{(p, \lcom, q) \mid 
  \forall \lcom', \lpcom \sothat \trans{\lcom}{\lpcom}{\lcom'} \implies 
      \exists p' \sothat \sat{\tid}{\lpcom}{p}{p'} 
          \land (p', \lcom', q) \in X} \\
& & \hfill {} \cup \dc{(p, \cskip, q) \mid p \rimpl q}
\end{array}
\]
Note that a powerset domain ordered by inclusion is a complete lattice and $F$
is a mapping on it, which means that $F$ is monotone. Consequently, by
Knaster-Tarski fixed-point theorem $F$ has the greatest fixed-point. It is easy
to see that $\sfsafe \triangleq \gfp{F_\tid}$ in Definition~\ref{def:safety}.

In the proof of Lemma~\ref{lem:safe} we use the following properties of the
action judgement and the $\rimpl$ relation.

\begin{prop}[Locality]\label{prop:axiomlocal}
\[
\begin{array}{l}
\forall p, q, r \sothat p \rimpl q \implies p * r \rimpl q * r;\\
\forall \tid, \lpcom, p, q, r \sothat \sat{\tid}{\lpcom}{p}{q} \implies
\sat{\tid}{\lpcom}{p * r}{q * r}.
\end{array}
\]
\end{prop}

\begin{prop}[Consequence]\label{prop:axiomconseq}
\[
\forall \tid, \lpcom, p, q, p', q' \sothat p' \rimpl p \land \sat{\tid}{\lpcom}{p}{q} \land 
q \rimpl q' \implies \sat{\tid}{\lpcom}{p'}{q'}.
\]
\end{prop}

The proofs of Propositions~\ref{prop:axiomlocal} and \ref{prop:axiomconseq} are straighforward:
both properties can be easily checked after unfolding definitions of action judgements.

\begin{prop}[Distributivity]\label{prop:axiomdist}
\[
\forall \tid, \lpcom, p_1, p_2, q_1, q_2 \sothat 
\sat{\tid}{\lpcom}{p_1}{q_1} \land \sat{\tid}{\lpcom}{p_2}{q_2} \implies
\sat{\tid}{\lpcom}{p_1 \lor p_2}{q_1 \lor q_2}.
\]
\end{prop}
\begin{proof}
  According to the Definition~\ref{def:axiom} of the action judgement $\sat{\tid}{\lpcom}{p_1 \lor
  p_2}{q_1 \lor q_2}$, in order to prove the latter we need to demonstrate the following:
\begin{multline}\label{eq:axiomdist1}
\forall r, \lstate, \lstate', \hstate, \tstate \sothat 
\lstate' \in \intp{\tid}{\lpcom}{\lstate} \land 
(\lstate, \hstate, \tstate) \in \reif{(p_1 \lor p_2) \vop r} \implies{} \\ 
\exists \hstate', \tstate' \sothat 
\lptrans{\hstate, \tstate}{\hstate', \tstate'} 
\land (\lstate', \hstate', \tstate') \in \reif{(q_1 \lor q_2) \vop r}.
\end{multline}
Let us consider any view $r$, states $\lstate, \lstate', \hstate$ and tokens
$\tstate$ such that both $\lstate' \in \intp{\tid}{\lpcom}{\lstate}$ and
$(\lstate, \hstate, \tstate) \in \reif{(p_1 \lor p_2) \vop r}$ hold. According
to the properties of disjunction stated in equalities (\ref{eq:disj}),
\[
\reif{(p_1 \lor p_2) \vop r} = \reif{(p_1 \vop r) \lor (p_2 \vop r)} =
  \reif{p_1 \vop r} \cup \reif{p_2 \vop r}.
\] 
Consequently, $(\lstate, \hstate, \tstate) \in \reif{p_1 \vop r} \cup \reif{p_2
\vop r}$.

Let us assume that $(\lstate, \hstate, \tstate) \in \reif{p_1 \vop r}$ (the
other case is analogous). Then according to the action judgement
$\sat{\tid}{\lpcom}{p_1}{q_1}$, there exist $\hstate'$ and $\tstate'$
such that:
\begin{equation}\label{eq:axiomdist2}
\lptrans{\hstate, \tstate}{\hstate', \tstate'} 
\land (\lstate', \hstate', \tstate') \in \reif{q_1 \vop r}.
\end{equation}
Once again, according to the properties (\ref{eq:disj}) of disjunction:
\[
\reif{q_1 \vop r} \subseteq \reif{q_1 \vop r} \cup \reif{q_2 \vop r} = 
    \reif{(q_1 \vop r) \lor (q_2 \vop r)} = \reif{(q_1 \lor q_2) \vop r},
\]
which together with (\ref{eq:axiomdist2}) means that $(\lstate', \hstate',
\tstate') \in \reif{(q_1 \lor q_2) \vop r}$. Overall we have shown that there
exist $\hstate'$ and $\tstate'$ such that $\lptrans{\hstate, \tstate}{\hstate',
\tstate'} $ and $(\lstate', \hstate', \tstate') \in \reif{(q_1 \lor q_2) \vop
r}$, which concludes the proof of (\ref{eq:axiomdist1}).
\end{proof}

\medskip

\noindent We now prove the closure properties from Lemma~\ref{lem:safe}.

\mypar{Proof of {\sc Frame}}
  Let us define an auxiliary function:
\[
\fbin{X}{r} \triangleq \dc{ (p*r, \lcom, q*r) \mid (p, \lcom, q) \in X}.
\]
Then our goal is to prove that $\fbin{\sfsafe}{r} \subseteq \sfsafe$. Since
$\sfsafe = \gfp{F_\tid}$, we can do a proof by coinduction: to conclude that
$\fbin{\sfsafe}{r} \subseteq \sfsafe$ holds, we demonstrate $\fbin{\sfsafe}{r}
\subseteq F_\tid(\fbin{\sfsafe}{r})$.

Consider any $(p', \lcom, q') \in \fbin{\sfsafe}{r}$. There necessarily are $p,
q, r$ such that $p' = p * r$, $q' = q * r$ and $(p, \lcom, q) \in
\sfsafe$. Let us assume that $\lcom = \cskip$.  Then $(p, \lcom, q) \in
\sfsafe$ implies that $p \rimpl q$. By Proposition~\ref{prop:axiomlocal}, $p *
r \rimpl q * r$, which implies $(p*r, \cskip, q*r) \in \sfsafe$ to hold.

Now let $\lcom \not= \cskip$. Since $(p, \lcom, q) \in \sfsafe$, by definition
of the safety relation the following holds of every $\lpcom$, $\lcom'$ and any
transition $\trans{\lcom}{\lpcom}{\lcom'}$:
  \[
  \exists p' \sothat \sat{\tid}{\lpcom}{p}{p'} \land (p', \lcom', q) \in \sfsafe.
  \]
By Proposition~\ref{prop:axiomlocal}, $\sat{\tid}{\lpcom}{p}{p'}$ implies
$\sat{\tid}{\lpcom}{p*r}{p'*r}$. Also, when $(p', \lcom', q) \in \sfsafe$, it
is the case that $(p'*r, \lcom', q * r) \in \fbin{\sfsafe}{r}$. Thus, we have
shown for every transition $\trans{\lcom}{\lpcom}{\lcom'}$ that there
exists $p'' = p' * r$ such that $\sat{\tid}{\lpcom}{p*r}{p''}$ and $(p'',
\lcom', q*r) \in \fbin{\sfsafe}{r}$:
  \[
  \forall \lpcom, \lcom \sothat \trans{\lcom}{\lpcom}{\lcom'} \implies \exists p'' \sothat 
             \sat{\tid}{\lpcom}{p*r}{p''} \land
             (p'', \lcom', q*r) \in \fbin{\sfsafe}{r},
  \]
  which is sufficient to conclude that $(p * r, \lcom, q * r) \in F_\tid(\fbin{\sfsafe}{r})$.
  \qed

\mypar{Proof of {\sc Choice}}
  Let us define an auxiliary function:
  \[
  \fbin{X}{Y} \triangleq \dc{ (p, \lcom_1 \ldet \lcom_2, q) \mid (p, \lcom_1, q) \in X 
                        \land (p, \lcom_2, q) \in Y}.
  \]
  Then our goal is to prove that $\fbin{\sfsafe}{\sfsafe} \subseteq
  \sfsafe$. For convenience, we prove an equivalent inequality $\fbin{\sfsafe}{\sfsafe} \cup
  \sfsafe \subseteq \sfsafe$ instead.

  Since $\sfsafe = \gfp{F_\tid}$, we can do a proof by coinduction: to conclude that
  $\fbin{\sfsafe}{\sfsafe} \cup \sfsafe \subseteq \sfsafe$ holds, we demonstrate
  $\fbin{\sfsafe}{\sfsafe} \cup \sfsafe \subseteq F_\tid(\fbin{\sfsafe}{\sfsafe} \cup
  \sfsafe)$. 

  Let us consider $(p, \lcom, q) \in \sfsafe$. Since $\sfsafe = \gfp{F_\tid}$, we
  know that $\sfsafe = F_\tid(\sfsafe)$ holds. Then by monotonicity of $F_\tid$, $(p,
  \lcom, q) \in F_\tid(\sfsafe) \subseteq F_\tid(\fbin{\sfsafe}{\sfsafe} \cup \sfsafe)$.

  Now let us consider $(p, \lcom, q) \in \fbin{\sfsafe}{\sfsafe}$. There necessarily are $\lcom_1$
  and $\lcom_2$ such that $\lcom = \lcom_1 \ldet \lcom_2$, $(p, \lcom_1, q) \in \sfsafe$, and $(p,
  \lcom_2, q) \in \sfsafe$.  For $(p, \lcom_1 \ldet \lcom_2, q)$ to belong to
  $F_\tid(\fbin{\sfsafe}{\sfsafe} \cup \sfsafe)$, the following has to be proven for every transition
  $\trans{\lcom_1 \ldet \lcom_2}{\lpcom}{\lcom'}$:
  \begin{equation}\label{eq:safechoice1} 
    \exists p' \sothat 
             \sat{\tid}{\lpcom}{p}{p'} \land
             (p', \lcom', q) \in \fbin{\sfsafe}{\sfsafe} \cup \sfsafe.
  \end{equation}
  According to the rules of the operational semantics (Figure~\ref{fig:opsem}), whenever
  $\trans{\lcom_1 \ldet \lcom_2}{\lpcom}{\lcom'}$, necessarily $\lpcom = \id$ and either
  $\lcom' = \lcom_1$ or $\lcom' = \lcom_2$. Let us assume that $\lcom' = \lcom_1$ (the other case
  is analogous). The action judgement $\sat{\tid}{\id}{p}{p}$ holds trivially. Knowing that $(p,
  \lcom_1, q) \in \sfsafe$, it is easy to see that (\ref{eq:safechoice1}) can be satisfied by
  letting $p' = p$. Consequently, $(p, \lcom_1 \ldet \lcom_2, q) \in F_\tid(\fbin{\sfsafe}{\sfsafe} \cup
  \sfsafe)$, which concludes the proof. \qed

\mypar{Proof of {\sc Disj}}
  Let
 $$
 \fun{X} \triangleq \dc{ (p_1 \lor p_2, \lcom, q_1 \lor q_2) \mid (p_1, \lcom,
   q_1) \in X \land (p_2, \lcom, q_2) \in X}.
$$ 
Then our goal is to prove that $\fun{\sfsafe} \subseteq \sfsafe$. Since $\sfsafe
= \gfp{F_\tid}$, we can do a proof by coinduction: to conclude that $\fun{\sfsafe}
\subseteq \sfsafe$ holds, we demonstrate $\fun{\sfsafe} \subseteq
F_\tid(\fun{\sfsafe})$.

  Let us consider $(p, \lcom, q) \in \fun{\sfsafe}$. Then there necessarily are $p_1, q_1, p_2$ and
  $q_2$ such that $p = p_1 \lor p_2$, $q = q_1 \lor q_2$, and $(p_1, \lcom, q_1), (p_2, \lcom, q_2)
  \in \sfsafe$. From the latter we get that for any $\lpcom, \lcom'$ and a transition
  $\trans{\lcom}{\lpcom}{\lcom'}$ the following holds:
  \[\begin{array}{l}
    \exists p'_1 \sothat \sat{\tid}{\lpcom}{p_1}{p'_1} \land (p'_1, \lcom', q_1) \in \sfsafe; \\
    \exists p'_2 \sothat \sat{\tid}{\lpcom}{p_2}{p'_2} \land (p'_2, \lcom', q_2) \in \sfsafe.
  \end{array}\]
  Then it is the case that $(p'_1 \lor p'_2, \lcom', q_1 \lor q_2) \in \fun{\sfsafe}$. Moreover,
  $\sat{\tid}{\lpcom}{p_1 \lor p_2}{p'_1 \lor p'_2}$ holds by Proposition~\ref{prop:axiomdist}.
  Thus, we have shown for every transition $\trans{\lcom}{\lpcom}{\lcom'}$ that there exists
  $p' = p'_1 \lor p'_2$ such that $\sat{\tid}{\lpcom}{p_1 \lor p_2}{p'}$ and $(p', \lcom', q_1 \lor
  q_2) \in \fun{\sfsafe}$:
  \[
    \forall \lpcom, \lcom \sothat \trans{\lcom}{\lpcom}{\lcom'} \implies 
      \exists p' \sothat 
             \sat{\tid}{\lpcom}{p_1 \lor p_2}{p'} \land
             (p', \lcom', q_1 \lor q_2) \in \fun{\sfsafe}.
  \]
  which is sufficient to conclude that $(p_1 \lor p_2, \lcom, q_1 \lor q_2) \in F_\tid(\fun{\sfsafe})$.
  \qed
  
\mypar{Proof of {\sc Iter}} 
  To do a proof by coinduction, we strengthen {\sc Iter} property as follows:
  \begin{multline}\label{eq:safeiter1}
  \forall p, \lcom \sothat ((p, \lcom, p) \in \sfsafe \implies 
    (p, \iter{\lcom}, p) \in \sfsafe) \land{} \\
      (\forall p_1, \lcom_1 \sothat (p_1, \lcom_1, p) \in \sfsafe \land
        (p, \lcom, p) \in \sfsafe \implies (p_1, \lcom_1 \lseq \iter{\lcom}, p) \in \sfsafe).
  \end{multline}

\ag{Can you simplify this proof by using the result for Seq as a black box?
  E.g., the last conjunct above follows from Seq.}


  Let us define auxilliary functions:
  \begin{gather*}
  \fun{X} \triangleq \dc{ (p, \iter{\lcom}, p) \mid (p, \lcom, p) \in X}\\
  \funx{X} \triangleq \dc{ (p_1, \lcom_1 \lseq \iter{\lcom}_2, p_2) \mid 
      (p_1, \lcom_1, p_2), (p_2, \lcom_2, p_2) \in X}.
  \end{gather*}
  Using them, we rewrite (\ref{eq:safeiter1}) as $\fun{\sfsafe} \cup \funx{\sfsafe} \subseteq
  \sfsafe$. Let $\xi = \{ (p, \cskip, p) \}$. It is easy to see that $\fun{\sfsafe} \cup
     \funx{\sfsafe} \cup \xi \subseteq \sfsafe$ is also an equivalent reformulation of
     (\ref{eq:safeiter1}), since $\xi \subseteq \sfsafe$ always holds.

  Since $\sfsafe = \gfp{F_\tid}$, we can do a proof by coinduction: to conclude that $\fun{\sfsafe} \cup
  \funx{\sfsafe} \cup \xi \subseteq \sfsafe$, we demonstrate $\fun{\sfsafe} \cup \funx{\sfsafe}
  \cup \xi \subseteq F_\tid(\fun{\sfsafe} \cup \funx{\sfsafe} \cup \xi)$. 

  Consider any $(p, \lcom, q) \in \xi$. Necessarily, $\lcom = \cskip$ and $q =
  p$. Note that $p \rimpl p$ always holds, which by definition of $F_\tid$ is
  sufficient for $(p, \cskip, p) \in F_\tid(\fun{\sfsafe} \cup \funx{\sfsafe}
  \cup
  \xi)$. Thus, $(p, \lcom, q) \in F_\tid(\fun{\sfsafe} \cup \funx{\sfsafe} \cup
  \xi)$.

  Consider any $(p, \lcom', q) \in \fun{\sfsafe}$. Necessarily, $p = q$ and
  there exists a sequential command $\lcom$ such that $\lcom' = \iter{\lcom}$
  and $(p, \lcom, p) \in \sfsafe$. We need to show that $(p, \iter{\lcom}, p)
  \in F_\tid(\fun{\sfsafe} \cup \funx{\sfsafe} \cup \xi)$. For the latter to 
  hold, by definition of $F_\tid$ it is sufficient that for every $\lpcom$,
  $\lcom''$ and a transition $\trans{\lcom}{\lpcom}{\lcom''}$ the
  following be true:
  \begin{equation}\label{eq:safeiter2}
      \exists p'' \sothat 
               \sat{\tid}{\lpcom}{p}{p''} \land
               (p'', \lcom'', p) \in \fun{\sfsafe} \cup \funx{\sfsafe} \cup \xi.  
  \end{equation}
  According to the operational semantics in Figure~\ref{fig:opsem}, when there
  is a transition $\trans{\lcom}{\lpcom}{\lcom''}$, necessarily $\lpcom =
  \id$ and either $\lcom'' = \cskip$ or $\lcom'' = \lcom \lseq \iter{\lcom}$.
  Let us assume that $\lcom'' = \cskip$. Since both $(p, \cskip, p) \in \xi$
  and $\sat{\tid}{\lpcom}{p}{p}$ always hold, it is easy to see that letting
  $p'' = p$ satisfies (\ref{eq:safeiter2}). Now let us turn to the case when
  $\lcom'' = \lcom \lseq \iter{\lcom}$. Note that $(p, \lcom \lseq
  \iter{\lcom}, p) \in \funx{\sfsafe}$ holds by definition of $\psi$. Thus, by
  letting $p'' = p$ we satisfy (\ref{eq:safeiter2}).

  Consider $(p_1, \lcom_0, p_2) \in \funx{\sfsafe}$. Necessarily, there exist $\lcom_1$ and
  $\lcom_2$ such that:
  \begin{equation}\label{eq:safeiter3pr}
  \lcom_0 = \lcom_1 \lseq \iter{\lcom}_2 \land (p_1, \lcom_1, p_2) \in \sfsafe \land
  (p_2, \lcom_2, p_2) \in \sfsafe.
  \end{equation}
  We need to show that $(p_1, \lcom_1 \lseq \iter{\lcom}_2, p_2) \in
  F_\tid(\fun{\sfsafe} \cup \funx{\sfsafe} \cup \xi)$. For the latter to hold, we
  need to prove the following for every $\lpcom$, $\lcom'$ and a transition
  $\trans{\lcom_1 \lseq \iter{\lcom}_2}{\lpcom}{\lcom'}$:
  \begin{equation}\label{eq:safeiter3}
      \exists p' \sothat 
               \sat{\tid}{\lpcom}{p_1}{p'} \land
               (p', \lcom', p_2) \in \fun{\sfsafe} \cup \funx{\sfsafe} \cup \xi.  
  \end{equation}
  According to the operational semantics in Figure~\ref{fig:opsem}, when there is a transition
  $\trans{\lcom_1 \lseq \iter{\lcom}_2}{\lpcom}{\lcom'}$, either of the following is true:
  \begin{itemize}
  \item there are $\lcom'_1$ and a transition $\trans{\lcom_1}{\lcom}{\lcom'_1}$ such that $\lcom' = \lcom'_1 \lseq \iter{\lcom}_2$;
  \item $\lcom_1 = \cskip$, $\lcom' = \lcom_2$ and $\lpcom = \id$.
  \end{itemize}
  Let us assume that the former is the case. From (\ref{eq:safeiter3pr}) we know that $(p_1, \lcom_1, p_2) \in \sfsafe$, so by definition of the safety relation we get that:
  \[
  \exists p'_1 \sothat 
               \sat{\tid}{\lpcom}{p_1}{p'_1} \land
               (p'_1, \lcom'_1, p_2) \in \sfsafe. 
  \]
  Consequently, $(p'_1, \lcom'_1 \lseq \iter{\lcom}_2, p_2) \in \funx{\sfsafe}$. Thus, by letting $p' = p'_1$ we can satisfy (\ref{eq:safeiter3}).

  Now let $\lcom_1 = \cskip$ and $\lpcom = \id$. From (\ref{eq:safeiter3pr}) we know that $(p_1,
  \cskip, p_2) \in \sfsafe$, meaning that necessarily $p_1 \rimpl p_2$. It is easy to see that $p_1
  \rimpl p_2$ holds if and only if so does $\sat{\tid}{\id}{p_1}{p_2}$. Knowing that $(p_2, \lcom_2,
  p_2) \in \sfsafe$, we can satisfy (\ref{eq:safeiter3}) by letting $p' = p_2$.
\qed

\ag{I think you can remove q' below. Just combine any safe triples agreeing on
  the middle predicate. If this doesn't cause too much disruption, perhaps you
  could implement this.}

\mypar{Proof of {\sc Seq}} Let
$$
\fbin{X}{q'} \triangleq \dc{ (p, \lcom_1 \lseq \lcom_2, q) \mid(p, \lcom_1, q') \in X \land
  (q', \lcom_2, q) \in X}.
$$
Then our goal is to prove that $\fbin{\sfsafe}{q'} \subseteq \sfsafe$. For
convenience, we prove an equivalent inequality $\fbin{\sfsafe}{q'} \cup \sfsafe
\subseteq \sfsafe$ instead.

  Since $\sfsafe = \gfp{F_\tid}$, we can do a proof by coinduction: to conclude that $\fbin{\sfsafe}{q'}
  \cup \sfsafe \subseteq \sfsafe$ holds, we demonstrate $\fbin{\sfsafe}{q'} \cup \sfsafe \subseteq
  F_\tid(\fbin{\sfsafe}{q'} \cup \sfsafe)$.

  Let us consider any $(p, \lcom, q) \in \sfsafe$. Since $\sfsafe = \gfp{F_\tid}$, we know that $\sfsafe
  = F_\tid(\sfsafe)$ holds. Then by monotonicity of $F_\tid$, $p, \lcom, q) \in F_\tid(\sfsafe) \subseteq
  F_\tid(\fbin{\sfsafe}{\sfsafe} \cup \sfsafe)$.

  Now let us consider any $(p, \lcom, q) \in \fbin{\sfsafe}{q'}$. There necessarily are $\lcom_1$
  and $\lcom_2$ such that:
  \begin{equation} \label{eq:safeseq1}
  \lcom = \lcom_1 \lseq \lcom_2 \land (p, \lcom_1, q') \in \sfsafe \land
    (q', \lcom_2, q) \in \sfsafe.
  \end{equation}
  For $(p, \lcom_1 \lseq \lcom_2, q)$ to belong to
  $F_\tid(\fbin{\sfsafe}{\sfsafe} \cup \sfsafe)$, the following has to be the case for every transition
  $\trans{\lcom_1 \lseq \lcom_2}{\lpcom}{\lcom'}$:
  \begin{equation}\label{eq:safeseqc1} 
    \exists p' \sothat \sat{\tid}{\lpcom}{p}{p'} \land
             (p', \lcom', q) \in \fbin{\sfsafe}{q'} \cup \sfsafe.
  \end{equation}
  According to the rules of the operational semantics (Figure~\ref{fig:opsem}), when there is a
  transition $\trans{\lcom_1 \lseq \lcom_2}{\lpcom}{\lcom'}$, either of the following is true:
  \begin{itemize}
  \item there exists $\lcom'_1$ such that $\lcom' = \lcom'_1 \lseq \lcom_2$ and
    $\trans{\lcom_1}{\lpcom}{\lcom'_1}$; or  
  \item $\lcom_1 = \cskip$, $\lpcom = \id$ and $\lcom' = \lcom_2$.
  \end{itemize}
  Let us assume that the former is the case. From (\ref{eq:safeseq1}) we know that $(p, \lcom_1,
  q') \in \sfsafe$, which means that the following holds of
  $\trans{\lcom_1}{\lpcom}{\lcom'_1}$:
  \[
  \exists p'' \sothat \sat{\tid}{\lpcom}{p}{p''} \land (p'', \lpcom'_1, q') \in \sfsafe.
  \]
  When $(p'', \lpcom'_1, q') \in \sfsafe$ and $(q', \lpcom_2, q) \in \sfsafe$, it is the case that
  $(p'', \lpcom'_1 \lseq \lpcom_2, q) \in \fbin{\sfsafe}{q'}$. Thus, by letting $p' = p''$ we
  satisfy (\ref{eq:safeseqc1}).

  We now consider the case when $\lcom_1 = \cskip$, $\lpcom = \id$ and $\lcom' = \lcom_2$. From
  (\ref{eq:safeseq1}) we know that $(p, \cskip, q') \in \sfsafe$, meaning that $p \rimpl q'$, or
  equivalently $\sat{\tid}{\id}{p}{q'}$. We also know from (\ref{eq:safeseq1}) that $(q', \lcom_2,
  q) \in \sfsafe$. Thus, (\ref{eq:safeseqc1}) can be satisfied by letting $p' = q'$.
  \qed

\mypar{Proof of {\sc Conseq}} Let us first show that $\safe{\tid}{p'}{\lcom}{q}$ holds, when so do
  $p' \rimpl p$ and $\safe{\tid}{p}{\lcom}{q}$. When $\lcom = \cskip$, $\safe{\tid}{p}{\lcom}{q}$
  gives us that $p \rimpl q$. It is easy to see that $p' \rimpl p$ and $p \rimpl q$ together imply
  $p' \rimpl q$, which is sufficient to conclude that $\safe{\tid}{p'}{\lcom}{q}$ holds. Let us
  assume that $\lcom \not= \cskip$. From $\safe{\tid}{p}{\lcom}{q}$ we get that the following holds
  of every transition $\trans{\lcom}{\lpcom}{\lcom'}$:
  \[
    \exists p'' \sothat \sat{\tid}{\lpcom}{p}{p''} \land \safe{\tid}{p''}{\lcom'}{q}
  \]
  However, by applying Proposition~\ref{prop:axiomconseq} about Consequence property of axiom
  judgements to $p' \rimpl p$ and $\sat{\tid}{\lpcom}{p}{p''}$ we get that
  $\sat{\tid}{\lpcom}{p'}{p''}$. Together with the formula above, it allows us to conclude that
  $\safe{\tid}{p'}{\lcom}{q}$ holds.

\ag{Why are you fixing q? Can't phi just weaken the postcondition in any safe
  triple in any way?}

  Now let us prove that $\safe{\tid}{p}{\lcom}{q'}$ holds, when so do $q \rimpl q'$ and
  $\safe{\tid}{p}{\lcom}{q}$. We define an auxilliary function:
  \[
    \fbin{X}{q} \triangleq \dc{ (p, \lcom, q') \mid (p, \lcom, q) \in X \land q \rimpl q'}.
  \]
  Our goal is to prove that $\fbin{\sfsafe}{q} \subseteq \sfsafe$. Since $\sfsafe = \gfp{F_\tid}$, we
  can do a proof by coinduction: to conclude that $\fbin{\sfsafe}{q} \subseteq \sfsafe$ holds, we
  demonstrate $\fbin{\sfsafe}{q} \subseteq F_\tid(\fbin{\sfsafe}{q})$.

  Let us consider any $(p, \lcom, q') \in \fbin{\sfsafe}{q}$. Necessarily, $(p, \lcom,
  q) \in \sfsafe$ and $q \rimpl q'$. When $\lcom = \cskip$, we need to show that $p \rimpl q'$.
  Since $(p, \cskip, q) \in \sfsafe$, it is the case that $p \rimpl q$. It is easy to see that
  $p \rimpl q$ and $q \rimpl q'$ together imply $p \rimpl q'$, which is sufficient to conclude
  that $(p, \lcom, q') \in F_\tid(\fbin{\sfsafe}{q})$.

  Now consider the case when $\lcom \not= \cskip$. Since $(p, \lpcom, q) \in \sfsafe$, by definition of the safety relation the following holds of
  every $\lpcom$, $\lcom'$ and a transition $\trans{\lcom}{\lpcom}{\lcom'}$:
  \[
  \exists p'' \sothat \sat{\tid}{\lpcom}{p}{p''} \land (p'', \lcom', q) \in \sfsafe
  \]
  Knowing that $q \rimpl q'$ and $(p'', \lcom', q) \in \sfsafe$, it is easy to see that $(p'',
  \lcom', q') \in \fbin{\sfsafe}{q}$. Thus, we have shown that:
  \[
  \forall \lpcom, \lcom' \sothat \trans{\lcom}{\lpcom}{\lcom'} \implies
    \exists p'' \sothat 
             \sat{\tid}{\lpcom}{p}{p''} \land
             (p'', \lcom', q') \in \fbin{\sfsafe}{q},
  \]
  which is sufficient for $(p, \lcom, q') \in F_\tid(\fbin{\sfsafe}{q})$ to hold.
\qed

\section{Proof of Lemma~\ref{lem:logic2safety}}

\mypar{Lemma}
    $\forall \tid, \lP, \lcom, \lQ \sothat \infer{\tid}{\lP}{\lcom}{\lQ} \implies 
        \forall \lint \sothat \safe{\tid}{\evalf{\lP}{\lint}}{\lcom}{\evalf{\lQ}{\lint}}.$

We prove Lemma~\ref{lem:logic2safety} by rule induction. For that we choose
arbitrary thread identifier $\tid$ and demonstrate that $\forall \lint \sothat
\safe{\tid}{\evalf{\lP}{\lint}}{\lcom}{\evalf{\lQ}{\lint}}$ is closed under the
proof rules from Figure~\ref{fig:proofrules}. The cases of {\sc Choice}, {\sc
Iter}, {\sc Seq}, {\sc Conseq}, {\sc Frame} and {\sc Disj} rules are
straightforward: they trivially follow from Lemma~\ref{lem:safe} after using
the properties of $\evalf{-}{\lint}$ from Figure~\ref{fig:sem4frm}. The {\sc
Ex} rule uses the fact that $\Vals$, which is the range of $i$, is finite,
which makes possible proving it just like the {\sc Disj} rule.

It remains to consider the {\sc Prim} rule to conclude
Lemma~\ref{lem:logic2safety}. Let us assume that $\forall \lint' \sothat
\sat{\tid}{\lpcom}{\evalf{\lP}{\lint'}}{\evalf{\lQ}{\lint'}}$ holds. We need to
demonstrate that so does $\forall \lint \sothat
\safe{\tid}{\evalf{\lP}{\lint}}{\lpcom}{\evalf{\lQ}{\lint}}$. To conclude that
the latter holds, according Definition~\ref{def:safety} we need to prove the
following for every $\lint$:
\begin{equation}\label{eq:l2s-prim}
\forall \lcom', \lpcom' \sothat
          \trans{\lpcom}{\lpcom'}{\lcom'} \implies \exists p' \sothat
          \sat{\tid}{\lpcom'}{\evalf{\lP}{\lint}}{p'} \land
          \safe{\tid}{p'}{\lcom'}{\evalf{\lQ}{\lint}}
\end{equation}
According to the operational semantics from Figure~\ref{fig:opsem}, the only
transition from a command $\lpcom$ is $\trans{\lpcom}{\lpcom}{\cskip}$.
Thus, in the formula above $\lpcom' = \lpcom$ and $\lcom' = \cskip$. Note that
$\safe{\tid}{\evalf{\lQ}{\lint}}{\cskip}{\evalf{\lQ}{\lint}}$ holds trivially.
Additionally, by our assumption,
$\sat{\tid}{\lpcom}{\evalf{\lP}{\lint'}}{\evalf{\lQ}{\lint'}}$ holds for any
$\lint'$. Consequently, it holds for $\lint' = \lint$. We conclude that by
letting $p' = \evalf{\lQ}{\lint}$ we satisfy (\ref{eq:l2s-prim}).

\qed

\section{Proof of Theorem~\ref{thm:lin}}

We further refer to the assumptions of Theorem~\ref{thm:lin} as a relation $\safelib(\llib, \hlib, \preVA, \postVA)$ defined as follows.

\begin{dfn}\label{def:safelib}
Given a concrete library $\llib$, an abstract library $\hlib$ and $\preVA, \postVA : \ThreadID
\times \HPCom \to \VAssn$, we say that a relation $\safelib(\llib, \hlib, \preVA, \postVA)$ holds
if and only if the following requirements are met:
\begin{enumerate}
\item $\dom(\llib) = \dom(\hlib)$;
\item $\forall \lint, \tid, \hpcom, \lstate, \hstate, \tstate, r \sothat (\lstate, \hstate, \tstate) \in \reif{\evalf{\preVA(\tid, \hpcom)}{\lint} \vop r} \implies \tstate(\tid) = \todo{\hpcom}$;
\item $\forall \lint, \tid, \hpcom, \lstate, \hstate, \tstate, r \sothat (\lstate, \hstate, \tstate) \in \reif{\evalf{\postVA(\tid, \hpcom)}{\lint} \vop r} \implies \tstate(\tid) = \done{\hpcom}$;
\item $\begin{multlined}[t]
\forall \lint, \tid, \hpcom, \hpcom', r, \tstate \sothat 
((\lstate, \hstate, \tstate [\tid : \todo{\hpcom}]) \in \reif{\evalf{\preVA(\tid, \hpcom)}{\lint} \vop r} \iff{}\\
(\lstate, \hstate, \tstate [\tid : \done{\hpcom'}]) \in \reif{\evalf{\postVA(\tid, \hpcom')}{\lint} \vop r}).
\end{multlined}$
\item $\begin{multlined}[t]
\forall m, a, r, \tid \sothat m \in \dom(\llib) \land a, r \in \Vals \land 
\tid \in \ThreadID \implies{} \\ 
\infer{\tid}{ \preVA(\tid, \hlib(m, a, r))}{ \llib(m, a, r) }{ \postVA(\tid, \hlib(m, a, r)) };
\end{multlined}$
\end{enumerate}
\end{dfn}

To strengthen the statement of Theorem~\ref{thm:lin} as necessary for its proof, we define an
auxilliary relation, a {\em thread pool invariant}. With this relation we establish a
correspondence between the information about LP in a thread $\tid$ from a given view $\vV{\tid}$
and sequential commands in a thread $\tid$ of a concrete thread pool $\tp$ and abstract thread pool
$\atp$.

\begin{dfn}\label{def:tpinv}
Given a concrete library $\llib$, an abstract library $\hlib$, predicates $\preVA, \postVA :
\ThreadID \times \HPCom \to \VAssn$, a concrete thread pool $\tp$, an abstract thread pool $\atp$,
a view $\vV{\tid}$ and an interpretation of logical variables $\lint$, we say that a thread pool
invariant $\tpinv_\tid(\lint, \tp, \atp, \vV{\tid}, \tstate)$ holds in a thread $\tid$ if and only
if the following requirements are met:
\begin{itemize}
\item if $\tp(\tid) = \idle$, then 
      $\atp(\tid) = \idle$ and $\vV{\tid} \rimpl \evalf{\postVA(\tid, \_)}{\lint}$, or
\item there exist $\lcom, r, m, a$ such that 
      $\tp(\tid) = (\lcom, r)$ and the following holds:
\begin{multline*}
  \safe{\tid}{\vV{\tid}}{\lcom}{\postVA(\tid, \hlib(m, a, r))} \land
  (
  (\tstate(\tid) = \todo{\hlib(m, a, r)} \land \atp(\tid) = (\hlib(m, a, r), r) ) \lor{} \\
  (\tstate(\tid) = \done{\hlib(m, a, r)} \land \atp(\tid) = (\cskip, r) )
  ).
\end{multline*}
\end{itemize}
\end{dfn}

Finally, analogously to Definition~\ref{def:hsem}, we write down formally a definition of the set
of histories of abstract libraries.
\begin{dfn}\label{def:ahsem}
  We define $\hsemn{n}{\hlib, \atp, \hstate}$ as a set of histories such that $\hsemn{0}{\hlib,
  \atp, \hstate} \triangleq \dc{ \emptytr }$ and:
\[
\begin{array}{rcl}
  \hsemn{n}{\hlib, \atp, \hstate} & \triangleq & 
    \{((\tid, \call{m}{a}) ::\hist) \mid m \in \dom(\hlib) \land
        \atp(\tid) = \idle \land{}  \\ & & \hfill 
        \exists r \sothat \hist \in \hsemn{n-1}{\hlib, \atp[\tid : (\hlib(m, a, r), r)], \hstate}\} \\ & &
        {} \cup 
    \{\hist \mid \exists \tid, \hpcom, \hstate', r \sothat
        \atp(\tid) = (\hpcom, r) \land \hstate' \in \intp{\tid}{\hpcom}{\hstate}
        \land{} \\ & & \hfill \hist \in \hsemn{n-1}{\hlib, \atp[\tid : (\cskip, r)], \hstate'}\}
        \\ & & {} \cup
    \{((\tid, \ret{m}{r}) ::\hist) \mid m \in \dom(\hlib) \land \atp(\tid) = (\cskip, r) \land{} \\
        & & \hfill \hist \in \hsemn{n-1}{\hlib, \atp[\tid : \idle], \hstate}\}
\end{array}
\]
We let $\hsem{\llib, \lstate} = \bigcup_{n \geq 0} \hsemn{n}{\hlib, (\lambda \tid \sothat \idle),
\hstate}$ denote the set of all possible histories of a library $\hlib$ that start from a state
$\hstate$.
\end{dfn}

We are now ready to prove Theorem~\ref{thm:lin}.

\mypar{Proof}
Let us consider any $\llib, \hlib, \preVA, \postVA$ such that $\safelib(\llib, \hlib, \preVA,
\postVA)$ holds. Let us explain how we strengthen the statement of the theorem in this proof. We
prove that $\forall n \sothat \phi(n)$ holds with $\phi(n)$ formulated as follows:
\begin{multline}\label{eq:thmphi}
\phi(n) = \forall \lint, \lstate, \hstate, \tstate, \tp, \atp \sothat 
  (\exists \vV{1}, \dots, \vV{\NThreads} \sothat
    \tpinvall{\lint}{\tp}{\atp}{\vV{\tidk}}{\tstate} \land{} \\
    (\lstate, \hstate, \tstate) \in \reif{\vopall_{\tid \in \ThreadID} \vV{\tid}})
    \implies
     \hsemn{n}{\llib, \tp, \lstate} \subseteq \hsem{\hlib, \atp, \hstate}.
\end{multline}

\noindent
Note that according to the semantics of the assertion language $\Assn$
(Figure~\ref{fig:sem4frm}):
\[
\evalf{\vopall_{\tidk \in \ThreadID} (\exists \hpcom \sothat \postVA(\tidk, \hpcom))}{\lint} =
\vopall_{\tidk \in \ThreadID} \evalf{(\exists \hpcom \sothat \postVA(\tidk, \hpcom))}{\lint}.
\]
With that in mind, it is easy to see that letting $\vV{\tidk} = \evalf{(\exists \hpcom \sothat
\postVA(\tidk, \hpcom))}{\lint}$ for all $\tidk \in \ThreadID$, $\tp = (\lambda \tid \sothat
\idle)$ and $\atp = (\lambda \tid \sothat \idle)$ in (\ref{eq:thmphi}) yields the formula:
\begin{multline*}
(\forall \lint, \lstate, \hstate, \tstate \sothat
    (\lstate, \hstate, \tstate) \in \reif{\evalf{\vopall_{\tidk \in \ThreadID} 
      (\exists \hpcom \sothat \postVA(\tidk, \hpcom))}{\lint}} 
\implies{} \\
\bigcup_{n \geq 0} \hsemn{n}{\llib, \lambda \tid \sothat \idle, \lstate} \subseteq 
\hsem{\hlib, \lambda \tid \sothat \idle, \hstate}),
\end{multline*}
which coincides with the statement of the theorem.

We prove $\forall n \sothat \phi(n)$ by induction on $n$. Let us take any $\lint, \lstate, \hstate,
\tstate, \tp$ and $\atp$, and consider $\vV{1}, \dots, \vV{\NThreads}$ such that the premisses of
$\phi(n)$ hold:
\begin{equation}\label{eq:thmpremiss}
\tpinvall{\lint}{\tp}{\atp}{\vV{\tidk}}{\tstate} \land (\lstate,\hstate, \tstate) \in
\reif{\vopall_{\tidk \in \ThreadID} \vV{\tidk}}
\end{equation}
We need to demonstrate that every history $h$ of the concrete library $\llib$ from the set
$\hsemn{n}{\llib, \tp, \lstate}$ is also a history of the abstract library $\hlib$: $h \in
\hsem{\hlib, \atp, \hstate}$.

By Definition~\ref{def:hsem} of $\hsemn{n}{\llib, \tp, \lstate}$, if $n = 0$, then $h$ is an empty
history that is trivially present in $\hsem{\hlib, \atp, \hstate}$. Let us now consider $n > 0$ and
assume that $\phi(n-1)$ holds. By definition of $\hsemn{n}{\llib, \tp, \lstate}$, $h$ corresponds
to one of the three events in a thread $\tid$: a call of an arbitrary method $m$ with an argument
$a$ in a thread $\tid$, a return from a method $m$ with a return value $r$ or a transition in a
thread $\tid$. We consider each case separately.

\mypar{Case \#1} There is a history $h'$, a thread $\tid$, a method $m \in \dom(\llib)$, its
argument $a$ and a return value $r$ such that $h = (\tid, \call{m}{a}) :: h'$, $\tp(\tid) = \idle$
and $h' \in \hsemn{n-1}{\llib, \tp[\tid : (\llib(m, a, r), r), \lstate]}$. By
Definition~\ref{def:ahsem}, to conclude that $h = (\tid, \call{m}{a}) :: h' \in \hsem{\hlib,
\atp, \hstate}$ it is necessary to show that $\atp(\tid) = \idle$ and $h' \in \hsem{\hlib,
\atp[\tid : (\hlib(m, a, r), r)], \hstate}$, which we further do in the proof of Case \#1.

According to (\ref{eq:thmpremiss}), $\tpinv_\tid(\lint, \tp, \atp, \vV{\tid}, \tstate)$ and
$(\lstate,\hstate, \tstate) \in \reif{\vopall_{\tidk \in \ThreadID} \vV{\tidk}}$ hold. Then
necessarily $\atp(\tid) = \idle$ and $\vV{\tid} \rimpl \evalf{\postVA(\tid, \_)}{\lint}$, which
corresponds to the only case when $\tp(\tid) = \idle$ in the thread pool invariant. By
Definition~\ref{def:rimpl} of $\vV{\tid} \rimpl \evalf{\postVA(\tid, \_)}{\lint}$,
$(\lstate,\hstate, \tstate) \in \reif{\vopall_{\tidk \in \ThreadID} \vV{\tidk}}$ implies $(\lstate,
\hstate, \tstate) \in \reif{\evalf{\postVA(\tid, \_)}{\lint} \vop \vopall_{\tidk \in \ThreadID
\setminus \dc{\tid}} \vV{\tidk}}$. From requirements to predicates $\preVA$ and $\postVA$ in
$\safelib(\llib, \hlib, \preVA, \postVA)$ we obtain that the following holds of $(\lstate, \hstate,
\tstate)$:
\begin{itemize}
\item $\tstate(\tid) = \done{\_}$, and
\item $(\lstate, \hstate, \tstate [\tid : \todo{\hlib(m, a, r)}]) \in 
  \reif{\evalf{\preVA(\tid, \hlib(m, a, r))}{\lint} \vop \vopall_{\tidk \in \ThreadID \setminus 
    \dc{\tid}} \vV{\tidk}}$.
\end{itemize}
Let $\vV{\tid}' = \evalf{\preVA(\tid, \hlib(m, a, r))}{\lint}$ and $\vV{\tidk}' = \vV{\tidk}$ for
all $\tidk \not= \tid$. Obviously, $(\lstate, \hstate, \tstate [\tid : \todo{\hlib(m, a, r)}]) \in
\reif{\vopall_{\tidk} \vV{\tidk}'}$.

Also, by Lemma~\ref{lem:logic2safety}, $\safe{\tid}{\evalf{\preVA(\tid, \hlib(m, a,
r))}{\lint}}{\llib(m, a, r)}{\evalf{\postVA(\tid, \hlib(m, a, r))}{\lint}}$ holds. This allows us
to conclude that in a thread $\tid$ a thread pool invariant $\tpinv_\tid(\lint, \tp[\tid :
(\llib(m, a, r), r)], \atp[\tid : (\hlib(m, a, r), r)], \vV{\tid}', \tstate[\tid : \hlib(m, a,
r)])$ holds. Moreover, according to (\ref{eq:thmpremiss}), thread pool invariants hold in all other
threads as well.

We have shown that there exist $\vV{1}', \dots, \vV{\NThreads}'$ such that:
\begin{multline*}
(\forall \tidk \sothat \tpinv_\tidk(\lint, \tp[\tid : (\llib(m, a, r), r)], 
  \atp[\tid : (\hlib(m, a, r), r)], \vV{\tidk}', \tstate[\tid : \hlib(m, a, r)])) \land{} \\
    (\lstate, \hstate, \tstate[\tid : \todo{\hlib(m, a, r)}]) \in 
      \reif{\vopall_{\tidk \in \ThreadID} \vV{\tidk}'},
\end{multline*}
which by the induction hypothesis $\phi(n-1)$ implies that $h' \in \hsemn{n-1}{\llib,
\tp[\tid : (\llib(m, a, r), r)], \lstate} \subseteq \hsem{\hlib, \atp[\tid : (\hlib(m, a, r), r)],
\hstate}$. We have also show that $\atp(\tid) = \idle$. By Definition~\ref{def:ahsem},
$h = (\tid, \call{m}{a}) :: h' \in \hsem{\hlib, \atp, \hstate}$, which concludes the proof of
Case~\#1.

\mypar{Case \#2} There is a history $h'$, a thread $\tid$, a method $m \in \dom(\llib)$, its
argument $a$ and a return value $r$ such that $h = (\tid, \ret{m}{r}) :: h'$, $\tp(\tid) = (\cskip,
r)$ and $h' \in \hsemn{n-1}{\llib, \tp[\tid : \idle], \lstate}$. By Definition~\ref{def:ahsem}, to
conclude that $h = (\tid, \ret{m}{r}) :: h' \in \hsem{\hlib, \atp, \hstate}$ it is necessary to
show that $\atp(\tid) = (\cskip, r)$ and $h' \in \hsem{\hlib, \atp[\tid : \idle], \hstate}$, which
we further do in this proof of Case \#2.

According to (\ref{eq:thmpremiss}), a thread invariant $\tpinv_\tid(\lint, \tp, \atp, \vV{\tid},
\tstate)$ holds. Then the following is true:
\begin{multline}\label{eq:thm2}
  \safe{\tid}{\vV{\tid}}{\cskip}{\evalf{\postVA(\tid, \hlib(m, a, r))}{\lint}} \land
  (
  (\tstate(\tid) = \todo{\hlib(m, a, r)} \land \atp(\tid) = (\hlib(m, a, r), r) ) \lor{} \\
  (\tstate(\tid) = \done{\hlib(m, a, r)} \land \atp(\tid) = (\cskip, r) )
  ).
\end{multline}

By Definition~\ref{def:safety} of $\safe{\tid}{\vV{\tid}}{\cskip}{\evalf{\postVA(\tid, \hlib(m, a,
r))}{\lint}}$, $\vV{\tid} \rimpl \evalf{\postVA(\tid, \hlib(m, a, r))}{\lint}$ holds. Consequently,
by Definition~\ref{def:rimpl}:
\[
(\lstate, \hstate, \tstate) \in 
\reif{\vV{\tid} \vop \vopall_{\tidk \in \ThreadID \setminus \dc{\tid}} \vV{\tidk}} \subseteq 
\reif{\evalf{\postVA(\tid, \hlib(m, a, r))}{\lint} \vop \vopall_{\tidk \in \ThreadID 
    \setminus \dc{\tid}} \vV{\tidk}}.
\]
From the third requirement to $\postVA$ in Definition~\ref{def:safelib}:
\[
(\lstate, \hstate, \tstate) \in \reif{\evalf{\postVA(\tid, \hlib(m, a, r))}{\lint} \vop
  \vopall_\tidk \vV{\tidk}} \implies \tstate(\tid) = \done{\hlib(m, a, r)}.
\]
Consequently, from (\ref{eq:thm2}) we get that $\tstate(\tid) = \done{\hlib(m, a, r)}$ and
$\atp(\tid) = (\cskip, r)$.

Let $\vV{\tid}' = \evalf{\postVA(\tid, \hlib(m, a, r))}{\lint}$ and $\vV{\tidk}' = \vV{\tidk}$ for
$\tidk \not= \tid$. It is easy to see that $\tpinv_\tid(\lint, \tp[\tid : \idle], \atp[\tid :
\idle], \vV{\tid}', \tstate)$ holds trivially by Definition~\ref{def:tpinv}. Moreover, according to
(\ref{eq:thmpremiss}), thread pool invariants hold in other threads as well.

We have shown that there exist $\vV{1}', \dots, \vV{\NThreads}'$ such that:
\[
(\forall \tidk \sothat \tpinv_\tidk(\lint, \tp[\tid : \idle], \atp[\tid : \idle], \vV{\tidk}', \tstate)) \land
  (\lstate, \hstate, \tstate) \in \reif{\vopall_{\tidk \in \ThreadID} \vV{\tidk}'},
\]
which by the induction hypothesis $\phi(n-1)$ implies that $h' \in \hsemn{n-1}{\llib, \tp[\tid :
\idle], \lstate} \subseteq \hsem{\hlib, \atp[\tid : \idle], \hstate}$. We have also show that 
$\atp(\tid) = (\cskip, r)$. By Definition~\ref{def:ahsem}, $h = (\tid, \ret{m}{r}) :: h' \in
\hsem{\hlib, \atp, \hstate}$, which concludes the proof of Case~\#2.

\mypar{Case \#3} There is a thread $\tid$, sequential commands $\lcom$ and
$\lcom'$, a primitive command $\lpcom$, concrete states $\lstate$ and
$\lstate'$ and a return value $r$ such that $\tp(\tid) = (\lcom, r)$,
$\sttrans{\lcom}{\lstate}{\tid}{\lpcom}{\lcom'}{\lstate'}$ and $h \in
\hsemn{n-1}{\llib, \tp[\tid : (\lcom', r)], \lstate'}$.

According to (\ref{eq:thmpremiss}), $\tpinv_\tid(\lint, \tp, \atp, \vV{\tid},
\tstate)$ holds. Consequently, there exist a method $m$ with its argument $a$
such that $\lcom = \hlib(m, a, r)$ and:
\begin{multline}\label{eq:thm3}
  \safe{\tid}{\vV{\tid}}{\lcom}{\evalf{\postVA(\tid, \hlib(m, a, r))}{\lint}}
  \land{}\\
  ( (\tstate(\tid) = \todo{\hlib(m, a, r)}
      \land \atp(\tid) = (\hlib(m, a, r), r) ) \lor{} \\
    (\tstate(\tid) = \done{\hlib(m, a, r)}
      \land \atp(\tid) = (\cskip, r) ) ).
\end{multline}
It is easy to see that whenever there is a transition
$\sttrans{\lcom}{\lstate}{\tid}{\lpcom}{\lcom'}{\lstate'}$, there also is a
stateless transition $\trans{\lcom}{\lpcom}{\lcom'}$. By
Definition~\ref{def:safety} of
$\safe{\tid}{\vV{\tid}}{\lcom}{\evalf{\postVA(\tid, \hlib(m, a, r))}{\lint}}$,
if $\trans{\lcom}{\lpcom}{\lcom'}$, then there exists a view $\vV{\tid}'$
such that $\sat{\tid}{\lpcom}{\vV{\tid}}{\vV{\tid}'}$ and
$\safe{\tid}{\vV{\tid}'}{\lcom'}{\evalf{\postVA(\tid, \hlib(m, a,
r))}{\lint}}$.

Let $\vV{\tidk}' = \vV{\tidk}$ for any $\tidk \not= \tid$. By
Definition~\ref{def:axiom} of the action judgement
$\sat{\tid}{\lpcom}{\vV{\tid}}{\vV{\tid}'}$, for $(\lstate, \hstate, \tstate)
\in \reif{\vV{\tid}\vop\vopall_{\tidk \not= \tid} \vV{\tidk}}$ and any
$\lstate' \in \intp{\tid}{\lpcom}{\lstate}$, there exist $\hstate', \tstate'$
such that:
\begin{equation}\label{eq:thm32}
\lptrans{\hstate, \tstate}{\hstate', \tstate'} 
\land (\lstate', \hstate', \tstate') \in 
    \reif{\vV{\tid}' \vop \vopall_{\tidk \not= \tid} \vV{\tidk}}.
\end{equation}

Let us assume that $\tstate = \tstate'$.
Note that $\safe{\tid}{\vV{\tid}'}{\lcom'}{\evalf{\postVA(\tid, \hlib(m, a,
r))}{\lint}}$ holds, and according to (\ref{eq:thm3}) the following holds too:
\begin{multline*}
(\tstate(\tid) = \todo{\hlib(m, a, r)} \land
  \atp(\tid) = (\hlib(m, a, r), r) ) \lor{} \\
  (\tstate(\tid) = \done{\hlib(m, a, r)} \land
    \atp(\tid) = (\cskip, r) )
\end{multline*}
Thus, it is easy to see that $\tpinv_\tid(\lint, \tp[\tid : (\lcom', r)], \atp,
\vV{\tid}', \tstate)$ holds. Combining this observation with (\ref{eq:thm32}),
we conclude that we have demonstrated existance of $\vV{1}', \dots,
\vV{\NThreads}'$ such that:
\[
(\forall \tidk \sothat
\tpinv_\tidk(\lint, \tp[\tid : (\lcom', r)], \atp, \vV{\tidk}', \tstate)) 
\land 
(\lstate', \hstate', \tstate) \in 
  \reif{\vopall_{\tidk \in \ThreadID} \vV{\tidk}'},
\]
which by the induction hypothesis $\phi(n-1)$ implies that $h \in
\hsemn{n-1}{\llib, \tp[\tid : (\lcom', r)], \lstate'} \subseteq \hsem{\hlib,
\atp, \hstate}$. This concludes the proof of the case when $\tstate =
\tstate'$.

We now return to the case when $\tstate \not= \tstate'$. According to
(\ref{eq:thm32}), $\lptrans{\hstate, \tstate}{\hstate', \tstate'}$ holds,
meaning that linearization points of one or more threads have been passed.
Without loss of generality, we assume the case of exactly one linearization
point, i.e. that $\lp{\hstate, \tstate}{\hstate', \tstate'}$ holds.
Consequently, according to Definition \ref{def:axiom} there exist $\tid'$ and
$\hpcom'$ such that:
\begin{equation}\label{eq:thm33}
\hstate' \in \intp{\tid'}{\hpcom'}{\hstate}
\land
\tstate(\tid') = \todo{\hpcom'}
\land
\tstate' = \tstate[\tid' : \done{\hpcom'}
\end{equation}

Let us consider the thread pool invariant $\tpinv_{\tid'}(\lint, \tp, \atp,
\vV{\tid'}, \tstate)$, which holds according to (\ref{eq:thmpremiss}). We show
that $\tp(\tid') \neq \idle$. From (\ref{eq:thm33}) we know that
$\tstate(\tid') = \todo{\hpcom'}$. Since the third requirement to $\postVA$ in
Definition~\ref{def:safelib} requires that $\tstate(\tid) = \done{\hlib(m, a,
r)}$ hold, by Definition~\ref{def:tpinv} it can only be the case that there
exist $\lcom'', m', a', r'$ such that $\tp(\tid') = (\lcom'', r')$ and the
following is true:
\begin{multline}\label{eq:thm34}
  \safe{\tid'}{\vV{\tid}}{\lcom''}{\evalf{\postVA(\tid', \hlib(m', a', r'))}{\lint}}
  \land{}\\
  ( (\tstate(\tid') = \todo{\hlib(m', a', r')}
      \land \atp(\tid') = (\hlib(m', a', r'), r') ) \lor{} \\
    (\tstate(\tid') = \done{\hlib(m', a', r')}
      \land \atp(\tid') = (\cskip, r') ) ).
\end{multline}
From formula (\ref{eq:thm33}) we know that $\hpcom' = \hlib(m', a', r')$.
Consequently, $\tstate'(\tid') = \done{\hlib(m', a', r')} \land \atp[\tid' :
(\cskip, r')](\tid') = (\cskip, r')$ holds, which allows us to conclude the
thread pool invariant $\tpinv_{\tid'}(\lint, \tp[\tid : (\lcom', r)],
\atp[\tid' : (\cskip, r')], \vV{\tid'}', \tstate')$ in case of $\tid' \neq
\tid$.

We now show that $\tpinv_{\tid}(\lint, \tp[\tid : (\lcom', r)],
\atp[\tid' : (\cskip, r')], \vV{\tid}', \tstate')$ hold, both when $\tid =
\tid'$ and $\tid \neq \tid'$. Let us first assume $\tid = \tid'$ ($r = r'$).
Then from (\ref{eq:thm3}) we get that $\hpcom = \hlib(m, a, r)$ and
$\tstate'(\tid) = \done{\hlib(m, a, r)}$ hold. When $\tid \neq \tid'$, no
abstract transition is made in $\tid$, so $\tstate(\tid) = \tstate'(\tid)$ and
$\atp(\tid) = \atp[\tid' : (\cskip, r')](\tid)$. Consequently, the following is
true in both cases:
\begin{multline}
  ( (\tstate'(\tid) = \todo{\hlib(m, a, r)}
      \land \atp[\tid' : (\cskip, r')](\tid) = (\hlib(m, a, r), r) ) \lor{} \\
    (\tstate'(\tid) = \done{\hlib(m, a, r)}
      \land \atp[\tid' : (\cskip, r')](\tid) = (\cskip, r') ) ).
\end{multline}
Together with $\safe{\tid}{\vV{\tid}'}{\lcom'}{\evalf{\postVA(\tid, \hlib(m, a,
r))}{\lint}}$ , those observations imply $\tpinv_\tid(\lint, \tp[\tid :
(\lcom', r)], \atp[\tid' : (\cskip, r')], \vV{\tid}', \tstate')$.

Combining these observations with $(\lstate', \hstate', \tstate') \in
\reif{\vV{\tid}' \vop \vopall_{\tidk \not= \tid} \vV{\tidk}}$ following from
(\ref{eq:thm32}), we conclude that we have demonstrated existance of $\vV{1}',
\dots, \vV{\NThreads}'$ such that:
\[
(\forall \tidk \sothat \tpinv_\tidk(\lint, \tp[\tid : (\lcom', r)], \atp[\tid' : (\cskip, r')], 
    \vV{\tidk}', \tstate')) \land
  (\lstate', \hstate', \tstate') \in \reif{\vopall_{\tidk \in \ThreadID} \vV{\tidk}'},
\]
which by the induction hypothesis $\phi(n-1)$ implies that $h \in
\hsemn{n-1}{\llib, \tp[\tid : (\lcom', r)], \lstate'} \subseteq \hsem{\hlib,
\atp[\tid : (\cskip, r')], \hstate'}$. Now that we demonstrated that
$\atp(\tid') = (\hlib(m', a', r'), r')$, $\hstate' \in \intp{\tid}{\hlib(m',
a', r')}{\hstate}$ and $h \in \hsem{\hlib, \atp[\tid : (\cskip, r')],
\hstate'}$ all hold, by Definition~\ref{def:ahsem} we can conclude that $h \in
\hsem{\hlib, \atp, \hstate}$.
\qed


\fi

\end{document}